\documentclass[onecolumn,prx, superscriptaddress,longbibliography,margin=1mm 11pt]{revtex4-2} 

\usepackage{algorithm}
\usepackage{ragged2e}
\usepackage{amssymb}
\usepackage{graphicx}
\usepackage{bm}
\usepackage{amsmath}
\usepackage{environ}
\usepackage{appendix}
\usepackage{physics}
\usepackage{rotating,booktabs,array,amsmath,amssymb}

\usepackage{algorithmicx,algorithm,setspace,subfigure}
\usepackage{float}

\usepackage{enumitem}

\newcolumntype{C}{>{$\displaystyle}c<{$}} 
\usepackage{caption}
\usepackage{subcaption}

\usepackage[nopar]{lipsum}

\usepackage{enumitem}
\setlist[enumerate,1]{label=(\arabic*),ref=\arabic*}

\makeatletter
\usepackage{amsthm}
\usepackage{amssymb}
\usepackage{amsfonts}
\usepackage{fancyhdr}
\usepackage{slashed}
\usepackage{bbm}
\usepackage{latexsym,epsfig,bbm}
\usepackage[colorlinks=true,urlcolor=blue,citecolor=blue]{hyperref}
\usepackage{todonotes}
\definecolor{blue}{rgb}{0,0.2,1}

\definecolor{red}{rgb}{0.9,0,0}

\newtheorem{theorem}{Theorem}
\newtheorem{lemma}[theorem]{Lemma}
\newtheorem{definition}[theorem]{Definition}

\renewcommand{\mathbf}[1]{\boldsymbol{#1}}

\begin{document}

\title{Fundamental questions on robustness and accuracy for classical and quantum learning algorithms}  \thanks{An invited book chapter (submitted June 2025) in \textit{Quantum Robustness in Artificial Intelligence -- Principles and Applications}, part of the Quantum Science and Technology book series, Springer, ed. Muhammad Usman, 2026}

\author{Nana Liu}
\email{nana.liu@quantumlah.org}
\affiliation{Institute of Natural Sciences, Shanghai Jiao Tong University, Shanghai 200240, China}
\affiliation{School of Mathematical Sciences, Shanghai Jiao Tong University, Shanghai, 200240, China}
\affiliation{Ministry of Education Key Laboratory in Scientific and Engineering Computing, Shanghai Jiao Tong University, Shanghai 200240, China}
\affiliation{Global College, Shanghai Jiao Tong University, Shanghai 200240, China.}

\date{\today}

\begin{abstract}
 This chapter introduces and investigates some fundamental questions on the relationship between accuracy and robustness in both classical and quantum classification algorithms under noisy and adversarial conditions. We introduce and clarify various definitions of robustness and accuracy, including corrupted-instance robustness accuracy and prediction-change robustness, distinguishing them from conventional accuracy and robustness measures. Through theoretical analysis and toy models, we establish conditions under which trade-offs between accuracy and robustness accuracy arise and identify scenarios where such trade-offs can be avoided. The framework developed highlights the nuanced interplay between model bias, noise characteristics, and perturbation types, including relevant and irrelevant perturbations. We explore the implications of some of these results for incompatible noise, adversarial quantum perturbations, the no free lunch theorem, and suggest future methods to examine these problems from the lens of dynamical systems. 
\end{abstract}

\maketitle
\section{Introduction}

Realizing the full potential of quantum-enhanced learning necessitates addressing a fundamental and pervasive challenge: the presence of noise and perturbations in both classical and quantum learning systems. Noise in machine learning arises from multiple sources. On one hand, environmental noise—such as hardware imperfections, decoherence, and stochastic fluctuations—can degrade the fidelity of quantum operations and corrupt input data. On the other hand, adversarial perturbations, often small but intentionally crafted changes to input examples, threaten the reliability and security of learning models by causing misclassifications. While noise is frequently viewed as a detrimental factor limiting performance, it can also play beneficial roles, including regularization to improve generalization, enhancement of privacy, and enabling compression when its parameters are carefully controlled.\\

In the context of quantum machine learning, noise plays a dual role: it can be a disruptive influence that impedes algorithmic performance, but also a tool to enhance learning capabilities when harnessed appropriately. This duality raises essential questions about how to evaluate and quantify the robustness of quantum classifiers under realistic noisy conditions. Conventional accuracy measures, defined in noise-free scenarios, fail to capture the stability of model predictions when inputs or models themselves are subject to perturbations. This necessitates the introduction of robustness-specific performance metrics designed for both adversarial and non-adversarial noise environments. \\

This chapter poses and explores some of these fundamental questions by focusing on classification problems under noisy conditions. We begin by systematically categorizing various notions of robustness, including corrupted-instance robustness accuracy, prediction-change robustness, and error-region robustness. We clarify how these definitions differ conceptually and mathematically, and how they relate to other accuracy measures. We also distinguish between relevant perturbations (which alter the true class of an example), and irrelevant perturbations, which do not change the underlying class but may affect model predictions. From some simple examples, we see the role of relevant and irrelevant perturbations trade-offs that can occur between accuracy and robustness accuracy.\\

The central theme of the chapter is the exploration of the trade-off relations between accuracy and robustness accuracy. Through theoretical analysis and simplified illustrative models, we demonstrate conditions under which increasing robustness inevitably sacrifices accuracy, and conversely, when both can be simultaneously optimized without compromise. Key results include general relationships for unbiased and biased classifiers, showing how robustness accuracy depends jointly on accuracy and model robustness properties.\\

We extend the discussion to quantum classification models, examining how quantum noise channels and measurement operators influence these trade-offs. Notably, we analyze noise models such as bit-flip, phase-flip, and depolarization, identifying scenarios where trade-offs appear or vanish. We also discuss adversarial perturbations modeled as quantum channels with bounded trace distance, establishing sufficient conditions for maximal trade-offs between accuracy and robustness accuracy. The chapter highlights how the intrinsic geometry of the data and the hypothesis class, rather than model-specific features alone, dictate these trade-offs.\\

Beyond these core theoretical developments, we explore important practical and conceptual applications of the accuracy-robustness trade-off framework:
\begin{itemize}
    \item Incompatible Noise: We introduce the notion of incompatible noise, where training a model to be robust against one type of noise can reduce robustness against another. We provide explicit examples involving depolarization and bit-flip noise to illustrate this phenomenon, which has significant implications for model design and training strategies.
\item	Adversarial Examples: We rigorously characterize adversarial perturbations in the quantum domain, specifying the conditions under which adversarial noise remains irrelevant (does not change true class) but induces prediction errors. This provides insight into the vulnerability of quantum classifiers and suggests avenues for robust training methods.
\item	No Free Lunch Theorem: The chapter offers a new perspective on the no free lunch theorem by linking it to robustness accuracy. We interpret the theorem’s assertion—that no universally superior model exists—through the lens of perturbations inducing shifts in data distributions, clarifying the relationship between accuracy on one distribution and robustness accuracy on perturbed distributions.
\end{itemize}
We also outline connections to dynamical systems theory, framing machine learning robustness as a question of stability and convergence under perturbations. We discuss the results and present the open questions in the last section. 

\section{Definition zoo}\label{sec:xref}
In studying the performance of models in the presence of environmental as well as adversarial noise, there are many different characterisations of this performance and the sensitivity to these perturbations. We have various measures of robustness, so it is important to first become clear what these different definitions are, as the literature can often confuse these different definitions and this can even lead to contradictory claims, which are in fact not contradictory if one is careful. 

\begin{table}[]
\caption{Classification of types of perturbations}
\begin{tabular}{|l|l|l|}
\hline
 & Irrelevant perturbation $\mathcal{N}_{irr}$ & Relevant perturbation $\mathcal{N}_{rel}$ \\ \hline
Small $D(\sigma, \mathcal{N}(\sigma))<\epsilon$ & Type I & Type II  \\ \hline
 Unrestricted $D(\sigma, \mathcal{N}(\sigma))$ & Type III & Type IV \\ \hline
\end{tabular}
\end{table}

\begin{table}[]
\caption{Different species of robustness and accuracy}
\begin{tabular}{|l|l|l|l|}
\hline
 & Corrupted-instance & Prediction-change & Error-region \\ \hline
Robust loss & $\tilde{L}$ \quad Definition 6 & $L^*$ \quad Definition 7 & $\bar{L}$ \quad Definition 8 \\ \hline
Robustness accuracy & $\tilde{A}$ \quad Definition 11 & N/A &  $\bar{A}$ \quad Definition 13\\ \hline
Robustness & N/A & $A^*$ \quad Definition 12  & N/A  \\ \hline
Robustness to adversarial perturbations & $\tilde{\Delta}$ \quad Definition 14 & $\Delta^*$ \quad Definition 15 & $\bar{\Delta}$ \quad Definition 16  \\ \hline
\end{tabular}
\end{table}

\subsection{General definitions for classification problems}
Let $\sigma \in \Sigma$ be input states, where the pair $(\sigma, c)$ is sampled i.i.d from some distribution.  The model or hypothesis function $h_{\vec \theta}:\Sigma \rightarrow C$ makes the prediction for the classification of states $\sigma \in \Sigma$ into $K$ distinct classes, where $C=\{1,...,K\}$. We denote the ground truth (true class) of $\sigma$ as $c(\sigma)$. If the model is deterministic, it means that $h$ is a deterministic function. If the model is probabilistic, then $h$ is a probabilistic function, which gives outputs in $C$ with a certain probability.\\

\begin{definition} \textbf{(Loss function)} A \textit{local loss function} $l_h(\sigma, c): (\Sigma, C) \rightarrow \Re$ is a penalty function for the state $\sigma$ when a model $h$ is used. It outputs a small number if $h$ gives the same prediction as $c$ for $\sigma$, and a larger number number otherwise. A \textit{global loss function} is the expected value of the local cost function over the distribution $(\sigma, c) \sim \mathcal{D}$
\begin{align}
    L_{h, \mathcal{D}}=\textbf{E}_{(\sigma, c) \sim \mathcal{D}} l_h(\sigma, c)
\end{align}
If $\mathcal{D}=\mathcal{D}_{train}$ corresponds only to the set of training data, then this global loss is called \textit{empirical loss}. If $\mathcal{D}=\mathcal{D}_{true}$ corresponds to the complete underlying distribution for $\sigma$ states, then this global loss is the \textit{expected loss} (population risk), but in most cases one does not have direct access to $\mathcal{D}_{true}$. 
\end{definition}

\begin{definition} \textbf{(Global loss minimisation)} If one found the model that minimised the empirical loss (training process), one gets the \textit{training loss} $\min_h L_{h, \mathcal{D}_{train}}$. The model obtained from this minimisation we denote $h_{train}$
\begin{align}
    h_{train}=\text{argmin}_h L_{h, \mathcal{D}_{train}}. 
\end{align}

One can also instead minimise the expected loss to get $\min_h L_{h, \mathcal{D}_{true}}$. The model obtained from this minimisation we denote $h_{true}$
\begin{align}
    h_{true}=\text{argmin}_h L_{h, \mathcal{D}_{true}}
\end{align}
where $h$ is optimised over some chosen hypothesis class of functions.
\end{definition}

\begin{definition} \textbf{(Generalisation and approximation error)} The \textit{generalisation error} (estimation error) $G_L$is defined as the difference
\begin{align}
    G_L=L_{h_{train}, \mathcal{D}_{true}}-L_{h_{true}, \mathcal{D}_{true}}.
\end{align}
Therefore, if the training set better approximates the true distribution (e.g. finite data versus infinite data setting), the the generalisation error should approach zero.
Note that while it is not in general possible to access $L_{h_{train}, \mathcal{D}_{true}}$ directly since $\mathcal{D}_{true}$ is generally unknown, it can be approximated by the \textit{testing loss} $L_{h_{train}, \mathcal{D}_{test}}$, where $\mathcal{D}_{test}$ is the distribution corresponding to the testing set. \\

The model minimising loss in classification is the Bayes classifier and we can call this value of loss the \textit{Bayes loss} $L_{h_{Bayes}, \mathcal{D}_{true}}$. Then the \textit{approximation error} $A_L$ refers to the difference 
\begin{align}
    A_L=L_{h_{true}, \mathcal{D}_{true}}-L_{h_{Bayes}, \mathcal{D}_{true}}.
\end{align}
The approximation error therefore captures how well the hypothesis class that the model is assumed to belong captures the true classification model. 
\end{definition}

\begin{definition} \textbf{Definition. (Excess risk)}
To assess how good the model $h_{train}$ is, that is derived from the training process, we are interested then in the \textit{excess risk} $E_L$
\begin{align}
    E_L=L_{h_{train}, \mathcal{D}_{true}}-L_{h_{Bayes}, \mathcal{D}_{true}}=G_L+A_L.
\end{align}
\end{definition}
Note that since $\mathcal{D}_{true}$ is not generally accessible, so the above quantity is approximated by replacing it with a test distribution $\mathcal{D}_{test}$.\\
\subsection{Definitions in presence of perturbations}
In the following, we want to deal with loss functions in the presence of perturbations to the input states $\sigma$. There are many definitions which must be clearly distinguished and can potentially have very different behaviour. \\

In the presence of a perturbation, the true class itself can change, depending on the type of perturbation on the test state. this leads us to the following definitions where we distinguish perturbations into two types.\\

\begin{definition} \textbf{(Relevant and irrelevant perturbations)} A \textit{relevant perturbation} $\mathcal{N}_{rel, \sigma}$ and \textit{irrelevant perturbation} $\mathcal{N}_{irr, \sigma}$ for the state $\sigma$ are changes to the state such that
\begin{align}
    &c(\mathcal{N}_{rel,\sigma}(\sigma)) \neq c(\sigma) \nonumber \\
    &c(\mathcal{N}_{irr,\sigma}(\sigma))=c(\sigma).
\end{align}
\end{definition}
We note that if the perturbation is artificially induced by an adversary, it is likely that there is $\sigma-$dependence in the noise $\mathcal{N}$. On the other hand, if the perturbations were induced as a result of the environment, it is more likely that the form of the perturbations themselves are independent of the state $\sigma$. \\

With respect to a specified perturbation, we can also define robust loss, and we identify three kinds.  
\begin{definition} \textbf{(Corrupted-instance robust loss)} Let the perturbation occur so that $\sigma \rightarrow \tilde{\sigma}$ for each state $\sigma \in \Sigma$. Then we define the \textit{corrupted-instance robust loss} (CI-robust loss) as the global loss $\tilde{L}$
\begin{align}
    \tilde{L} \equiv \textbf{E}_{(\sigma, c) \sim \mathcal{D}} l_h (\tilde{\sigma}, c_{\sigma})
\end{align}
\end{definition}
\begin{definition} \textbf{(Prediction-change robust loss)} Let the perturbation occur so that $\sigma \rightarrow \tilde{\sigma}$ for each state $\sigma \in \Sigma$. Then we define the \textit{prediction-change robust loss} (PC-robust loss) as the global loss $L^*$
\begin{align}
    L^* \equiv \textbf{E}_{(\sigma, c) \sim \mathcal{D}} l_h (\tilde{\sigma}, h(\sigma)),
\end{align}
where we replace the true class $c_{\sigma}$ with the model prediction $h(\sigma)$. 
\end{definition}
\begin{definition} \textbf{(Error-region robust loss)}
We define the \textit{error-region robust loss} (ER-robust loss) as the global loss $\bar{L}$
\begin{align}
    \bar{L} \equiv \textbf{E}_{(\sigma, c) \sim \mathcal{D}} l_h (\tilde{\sigma}, c_{\tilde{\sigma}})
\end{align}
where we also take into account the change to the true class if the perturbation $\sigma \rightarrow \tilde{\sigma}$ occurs. 
\end{definition}
However, since the true class is usually not accessible (at least those not in the training set), it is not convenient to work with the error-region robust loss, and instead it is more convenient to work with the corrupted-instance robust loss.

\subsubsection{$1-0$ loss: accuracy and robustness}\label{sec:10definitions}
Here we can focus on the different definitions of robustness and accuracy when using $1-0$ loss, which is an important loss function for classification problems. Here the local loss $l(\sigma, c)$ is the identity function $\mathbf{1}(h \neq c)$, which takes the value of $1$ if $h(\sigma) \neq c$ and is equal to $0$ otherwise.
\begin{definition}
The \textit{error region} $\mathcal{M}$ is defined as the set of samples $\sigma \in \Sigma$ where $h(\sigma)\neq c(\sigma)$. 
\end{definition}

\begin{definition} We define \textit{accuracy} of the model $h_{\vec \theta}$ as the probability 
\begin{align}
   A_{\vec \theta, \mathcal{D}} \equiv P_{{(\sigma, c) \sim \mathcal{D}}}(h_{\vec \theta} (\sigma)=c(\sigma)).
\end{align}
We will often suppress the subscripts $\{\vec \theta, \mathcal{D}\}$ in the cases where our model and distribution do not change in our analysis. 
\end{definition}

\begin{definition} We define \textit{corrupted-instance robustness accuracy} by the probability 
\begin{align}
     \tilde{A}_{\vec \theta, \mathcal{D}} \equiv P_{{(\sigma, c) \sim \mathcal{D}}}(h_{\vec \theta} (\tilde{\sigma})=c(\sigma))
\end{align}
where $\tilde{h}$ is a model modified by some perturbation $\mathcal{N}$.
\end{definition}

\begin{definition} We define \textit{prediction-change robustness} by the probability 
\begin{align}
     A^*_{\vec \theta, \mathcal{D}} \equiv P_{{(\sigma, c) \sim \mathcal{D}}}(h_{\vec \theta} (\tilde{\sigma})=h(\sigma)). 
     \end{align}
\end{definition}
Although we will not be using the following definition in our analysis later -- since corrupted-instance robustness accuracy and prediction-change robustness are more natural -- for completeness we include also the error-region robustness accuracy. 
\begin{definition} 
We define \textit{error-region robustness accuracy} by the probability 
\begin{align}
     \bar{A}_{\vec \theta, \mathcal{D}} \equiv P_{{(\sigma, c) \sim \mathcal{D}}}(h_{\vec \theta} (\tilde{\sigma})=\tilde{c}(\sigma)). 
\end{align}
\end{definition}
In the rest of the chapter when using the terminology robustness accuracy, we use this to denote corrupted-instance robustness accuracy unless otherwise specified. 
\subsubsection{Measures of robustness with respect to adversarial perturbations}
Note that all the above measures and definitions require a specification of what the perturbation is. However, in adversarial perturbations where by definition we do not necessarily have full information on the perturbation itself, and in addition these are usually 'small perturbations' according to some definition of 'small', it is important to characterise the robustness of a model. In this case, it is more convenient to specify some minimum distance -- again different possible definitions of distance are possible -- to perturb a test state whereby the relationship between the model prediction and the true prediction changes. Just as there are three types of robust loss, here we have three corresponding types of robustness with respect to adversarial perturbations. We note that depending on the definition of distance one chooses, for example $l_1$ or quantum trace distance, we will also obtain different values and can even obtain different behaviour. 
\begin{definition}
    We define \textit{corrupted-instance robustness to adversarial perturbations} by 
\begin{align}
   \tilde{\Delta}(h,\sigma)=\min {|\delta|} \quad, \quad h(\mathcal{N}_{\delta}(\sigma))\neq c(\sigma),
\end{align}
where $|\delta|$ is some measure of distance between states $\sigma$ and $\mathcal{N}_{\delta}(\sigma)$. 
We can also define the corresponding \textit{expected corrupted-instance robustness to adversarial perturbations} by
\begin{align}
   \langle  \tilde{\Delta}(h) \rangle=\textbf{E}_{{(\sigma, c) \sim \mathcal{D}}} \Delta(h,\sigma) 
\end{align}
\end{definition}
\begin{definition}
 We define \textit{prediction-change robustness to adversarial perturbations} by 
\begin{align}
    \Delta^*(h,\sigma)=\min {|\delta|} \quad, \quad h(\mathcal{N}_{\delta}(\sigma))\neq h(\sigma),
\end{align}
where $|\delta|$ is some measure of distance between states $\sigma$ and $\mathcal{N}_{\delta}(\sigma)$.
We can also define the corresponding \textit{expected prediction-change robustness to adversarial perturbations} by
\begin{align}
    \langle \Delta^*(h)\rangle =\textbf{E}_{{(\sigma, c) \sim \mathcal{D}}} \Delta(h,x) 
\end{align}
\end{definition}

\begin{definition}
We define \textit{error-region robustness to adversarial perturbations} by 
\begin{align}
    \bar{\Delta}(h,\sigma)=\min {|\delta|} \quad, \quad h(\mathcal{N}_{\delta}(\sigma))\neq c(\mathcal{N}_{\delta}(\sigma)),
\end{align}
where $|\delta|$ is some measure of distance between states $\sigma$ and $\mathcal{N}_{\delta}(\sigma)$.
We can also define the corresponding \textit{expected error-region robustness to adversarial perturbations} by
\begin{align}
    \langle \bar{\Delta}(h) \rangle= \textbf{E}_{{(\sigma, c) \sim \mathcal{D}}} \Delta(h,x) 
\end{align}
\end{definition}

\section{Trade-offs between accuracy and robustness accuracy}
In this section, we will examine more closely - through simple examples - the relationship between robustness and accuracy, and robustness accuracy from the point of view of classification algorithms at the testing stage. We will use some simplified toy models to illustrate some key observations and to show why it's insufficient to study accuracy of a model on its own. We will see through some simple examples that trade-offs between accuracy and robustness accuracy can exist for both the $1-0$ loss as well as linear loss. From simple examples as a starting point, we can ask fundamental questions about the requirements we should have for our model which could be subject to noisy perturbations and when we do and do not expect trade-offs between accuracy and robustness accuracy. \\

For discussions on trade-offs between accuracy and robustness for classification of classical data (where both accuracy and robustness may be defined in various ways), see for example~\cite{dobriban2023provable, tsipras2018robustness, xu2012robustness, yang2020closer, zhang2019theoretically}.

\subsection{$1-0$ loss}
From our definition zoo in Section~\ref{sec:10definitions} focusing on $1-0$ loss, we see that the main players in characterising how good the classifier are: (i) accuracy $A$, (ii) corrupted-instance robustness accuracy $\tilde{A}$ and (iii) prediction-change robustness $A^*$. A fundamental question is, what is the relationship between these key players? \\

Briefly, the accuracy denotes the probability that the model $h$ makes the correct description when we are sampling test states from the full distribution $\mathcal{D}$. However, accuracy may not be sufficient if the model is not trained under particular types of noise that is expected to corrupt incoming new states. The corrupted-instance robustness accuracy refers to the accuracy of the model when each test state $\sigma$ is perturbed by some perturbation $\mathcal{N}$, and we refer this as \textit{robustness accuracy} for short. Since this quantity itself must be related to properties of the model itself, we must also define prediction-change robustness --- which we refer as \textit{robustness} for short -- which refers to the probability that the model prediction itself does not change when each test state is subject to the same perturbation $\mathcal{N}$. \\

Here it is possible to identify a fundamental relationship between these quantities when considering $K$-class classification. Here we look at the simplest case where both the data and the model are unbiased, and we will make some simple observations. 

\begin{theorem} \textbf{(Unbiased classifier)} \label{theorem1}
\textit{If the data and the model are both unbiased, in the sense that $P(h=i)=1/K=P(c=j)$ for $K$-class classification and for any $i,j=1,...,K$, then the following holds:
\begin{align} \label{eq:lemma1}
     \tilde{A} &= A(2A^*-1)+(1-A^*) \\
     &=A^*(2A-1)+(1-A).
\end{align}}
\end{theorem}
\begin{proof}
    See Appendix A for the proof. 
\end{proof}
From this simple relationship, we can already make some important observations:
\begin{enumerate}
    \item It is insufficient to only have a high accuracy $A$. The aim for a good model is to have high robustness accuracy $\tilde{A}$ -- which depends just as much on robustness $A^*$ as on accuracy $A$. Here $A$ and $A^*$ are interchangeable in terms of their effect on $\tilde{A}$. 
    \item Trade-offs between  $\tilde{A}$ and $A$ are possible. For example, if model has low robustness $A^*<1/2$, higher accuracy leads to lower robustness accuracy, so training more won't help if the model itself is badly chosen. So the choice of the model itself and having in-built invariances to obtain high $A^*$ is very important. Similarly, if accuracy is low $A<1/2$, increasing robustness alone cannot achieve high robustness accuracy: in fact increasing robustness decreases robustness accuracy. 
    \item Theorem 1 can be rewritten $|\tilde{A}-A|=|(1-2A)(1-A^*)|$, where it can be seen that the difference between accuracy and robustness accuracy depends independently on $A$ and $A^*$. 
\end{enumerate}
In the above, it is clear that how well the model $h$ is chosen and its own robustness properties is vital to the relationship between accuracy and robustness accuracy. We note that this relationship can also be modified from the above simplest form when different assumptions are included. A simple example we can show explicitly is when the model itself is biased. 

\begin{theorem} \textbf{(Biased classifier))} \label{theorem2}
\textit{If the data is unbiased, in the sense that $1/K=P(c=j)$ for $K$-class classification, but the model is biased. We look at the simplest setting where only one class is biased. Suppose it is class $a$ that is biased, so then $P(h=a)=\alpha$, and the rest of the classes are unbiased, so $P(h=i \neq a)=(1-\alpha)/(K-1)$ for any $i \neq a$. Then the following holds:
\begin{align} \label{eq:lemma2}
  & \tilde{A}=\frac{(1-2A)(K-1)}{(1-\alpha)K} 
\left(A^*+\left(\frac{1-\alpha}{\alpha (K-1)}-1\right)P(\tilde{h}=h=a)\right)    +(1-A)
\end{align}}
\end{theorem}
\begin{proof}
    See Appendix B for the proof. 
\end{proof}
These simple observations show the importance of distinguishing robustness and robustness accuracy and accuracy as three independent considerations that must be studied in their own right, as well as their dependencies upon each other. The fact that trade-offs between these quantities can exist also points out applications of this to investigating whether it is possible to train models under 'incompatible noises', where more training under one class of noise can actually adversely affect robustness with respect to another type of noise. \\

Another application could be to the no free lunch theorem, which is based on the observation that for every distribution for which a model performs well, there should be some distribution for which the same model would perform badly. Although the free lunch theorem guarantees that these 'dual distributions' should exist, what is the relationship between these distributions and how do we identify them? These questions are also intimately connected with the trade-off relationships between accuracy, robustness accuracy and robustness that we see here. 

\subsubsection{Examples of trade-off between robustness accuracy and accuracy}
In this section, we provide some examples both in classical and quantum models where trade-offs between accuracy and robustness accuracy can exist. By this, we mean we are provided with two different models $h_1$ and $h_2$ and the test states are selected from an identical distribution $\mathcal{D}$. They have corresponding accuracies $A_1$, $A_2$, robustness accuracies $\tilde{A}_1$, $\tilde{A}_2$ and robustness values $A^*_1$, $A^*_2$. A trade-off between robustness accuracy and accuracy is said to exist if when $A_1<A_2$ we simultaneously have $\tilde{A}_1 < \tilde{A}_2$. See Appendix C for more details on the individual examples. Example 1 below is from \cite{tsipras2018robustness}. \\ 

\noindent \textbf{Example 1.}
    Suppose we have the following two models for classical data $x=(x_1,...,x_{d+1})$ that is selected from a distribution $D$, 
\begin{align}
    & H_1(x)=\frac{1}{d}(x_2+...+x_{d+1})\nonumber \\
    & H_2(x)=x_1.
\end{align}
This model is related to the true class in the following way. Let $h(x)=\text{sign}(H(x))$. Then $c(x)=h_2(x)$ with probability $p$ (i.e. proportion $p$ of states selected from $x \sim D$ satisfy $c(x)=h_2(x)$). Also if $c(x)=\pm 1$, then $(x_2,...,x_{d+1}) \sim D_{\pm}$. Suppose we have a noise model where we have the perturbation $x_2,...,x_{d+1} \rightarrow x'_2,...,x'_{d+1}$, where if  $(x_2,...,x_{d+1}) \sim D_{\pm}$ then $(x'_2,...,x'_{d+1}) \sim D_{\mp}$, and $x_1' \rightarrow x_1$. \\

Then when $p>0.5$, it is possible to have the tradeoff $A_1 \gg A_2$ and $\tilde{A}_1 \ll \tilde{A}_2$, so a more accurate model $H_1$ can lead to worse robustness accuracy. \\

\noindent \textbf{Example 2.}
    We can generalise Example 1 to a more general example. Suppose we have two models 
\begin{align}
    & H_1(x)=f(x_1, x_2,...,x_{d+1}) \nonumber \\
    & H_2(x)=g(x_1),
\end{align}
where the perturbation is only on $x_2,...,x_{d+1}$, leaving $x_1$ invariant, and they are selected from one of two distributions $D_{\pm}$ coinciding with $c(x)=\pm$. The effect of the perturbation is to change $D_{\pm} \rightarrow D_{\mp}$. Here again we have $A_2=p=\tilde{A}_2$. Then if $A_1>A_2=p$, it is always true that $\tilde{A}_1<\tilde{A}_2$ if $p>1/2$. \\ 

\noindent \textbf{Example 3.}
    Suppose we have two quantum models $H_1$ and $H_2$ and a Hamiltonian $h\otimes \mathbf{1}^{\otimes d}+\mathbf{1}\otimes h \otimes \mathbf{1}^{\otimes d-1}+...+\mathbf{1}^{\otimes d} \otimes h$ with local operator $h$. Let the input state be of the form $\sigma=\sigma_1 \otimes...\otimes \sigma_{d+1}$ where $\sigma_{i=1,...,d+1}$ are local states that $h$ acts on and 
    \begin{align}
   & H_1(\sigma)=\frac{1}{d}\sum_{i=1}^{d} \text{Tr}(h \sigma_i) \nonumber \\
   & H_2(\sigma)=\text{Tr}(h \sigma_{d+1})
\end{align}
and each $-1 \leq \text{Tr}(h \sigma_i) \leq 1$. 
We can define the distribution from which we select $\sigma$ as $D$. These models are related to the true class in the following way. Let $h(\sigma)=\text{sign}(H(\sigma))$ such that $c(\sigma)=h_2(\sigma)$ with probability $p$. Also if $c(\sigma)=\pm 1$, then the values $\text{Tr}(h \sigma_i) \vert_{i=1,...,d} \sim D_{\pm}$ \\

Now we choose a noise model such that the effect of the perturbation is $\text{Tr}(h \sigma_i) \vert_{i=1,...,d} \rightarrow \text{Tr}(h \sigma'_i) \vert_{i=1,...,d} \sim D_{\mp}$. \\

Then it is possible to find a value $p>0.5$ where we have the trade-off situation that higher accuracy $A_1>A_2$ implies less robustness accuracy $\tilde{A}_1<\tilde{A}_2$.\\

\noindent \textbf{Example 4.}
    We can look at more general models of the following form. We first define 
    \begin{align}
        f_i(\sigma)=\text{Tr}(\mathcal{O}_i \sigma)
    \end{align}
as the $i^{\text{th}}$ \textit{quantum feature with respect to $\mathcal{O}$} of the state $\sigma$, where $\{\mathcal{O}_i\}_{i=1}^{d+1}$ is a set of hermitian operators. Suppose the elements $\alpha_i=0,1$ for simplicity. Then we can compare two models 
\begin{align}
    & H_1(\sigma)=f_1(\sigma) \nonumber \\
    & H_2 (\sigma)=\sum_{i=1}^{d+1} \alpha_i f_i(\sigma)
\end{align}
and we are selecting states $\sigma$ in a way such that $f_{1}(\sigma)=\pm 1$ and $f_{2\leq i \leq d+1}(\sigma)$ has two distributions $D_{\pm}$ corresponding to $c(\sigma)=\pm$. Suppose we choose perturbations where only $f_2,...,f_{d+1}$ can change and leave $f_1$ invariant. Furthermore, the distributions change according to $D_{\pm} \rightarrow D_{\mp}$. We then define the accuracy $A_1$ when $P(c=1)=1/2=P(c=-1)$ to be $A_1=P(h_1=c) \equiv p$. A similar argument can be followed where when $p>0.5$, one can achieve $A_2>A_1$ while $\tilde{A}_2<\tilde{A}_1$. \\

The above examples 1-4 focuses on deterministic classifiers. We can also examine examples with probabilistic classifiers. In the case of a probabilistic classifier, the predicted label $h_s(\sigma)$ is a stochastic variable instead of a deterministic one given $\sigma$. For instance, we can consider this label as the outcome of a single-shot quantum measurement. We focus here on binary classifiers which can be straightforwardly generalised to the multiclass scenario. \\

However, the true label $c$ is still deterministic. This means we are sampling the non-stochastic pair $\{\sigma_i, c_i\} \sim \mathcal{D}$ from some underlying distribution $\mathcal{D}$, where $i$ can also range over continuous values. For simplicity of notation, we consider discrete $i$ and use the summation sign instead of integrals. \\

Below we will distinguish between two types of probabilities. $P_s$ is the probability associated with the stochastic variable getting a specific outcome given some $\sigma$, which we can model with a quantum expectation value as $P_s(h_s(\sigma)=j)=\text{Tr}(\Pi_j \sigma)$, where $\{\Pi_j\}$ are quantum measurement operators, e.g. POVMs. On the other hand $P$ is the probability associated with averaging over some distribution of input states, for instance $P(c(\sigma)=1)$ is the probability of states sampled from $\mathcal{D}$ as having the true label $1$. We can define then accuracy $A_s$ of this stochastic classifier $h_s$ as the probability it gives the same outcome as the true label
\begin{align}
    A_s \equiv \sum_{(\sigma_i, c_i)}P(\sigma_i, c_i)P_s(h_s(\sigma_i)=c_i). 
\end{align}
Then we the following lemma for an unbiased binary stochastic classifier. Here we use the definition $\sigma^{\pm}=\{\sigma| c(\sigma)=\pm 1\}$ and we can define the normalised quantum states 
\begin{align}
     \Sigma^{\pm} \equiv \frac{\sum_{\sigma_i^{\pm}}P(\sigma^{\pm}_i)\sigma_i^{\pm}}{\sum_{\sigma_j^{\pm}}P(\sigma^{\pm}_j)}. 
\end{align}
\begin{lemma}
    We assume an unbiased dataset $\sum_{\sigma_j^{\pm}}P(\sigma^{\pm}_j)=1/2$. We assume a noise model that modifies the measurement $\Pi_{1, -1} \rightarrow \tilde{\Pi}_{1, -1}=\sum_l E_l \Pi_{1, -1} E_l^{\dagger}$ where $\{E_l\}$ are Kraus operators and $\Pi_1=\mathbf{1}-\Pi_{-1}$. Then the difference between the robustness accuracy and accuracy is given by 
  \begin{align} \label{eq:asdifference1}
    |\tilde{A}_s-A_s|=\frac{1}{2}|\text{Tr}(\Pi_1(\tilde{\Sigma}^+-\Sigma^++\Sigma^--\tilde{\Sigma}^-)|.
\end{align}
We can also provide an upper bound to this difference
\begin{align}
    |\tilde{A}_s-A_s|\leq \tau(\Sigma_A, \Sigma_B)
\end{align}
where $\tau(\Sigma_A, \Sigma_B)$ is the trace distance between states $\Sigma_A \equiv (1/2)(\tilde{\Sigma}^++\Sigma^-)$ and $\Sigma_B \equiv (1/2)(\Sigma^++\tilde{\Sigma}^-)$. 
\end{lemma}
\begin{proof}
See Appendix D. 
\end{proof}
Note that furthermore, the projector $\Pi_1$ that \textit{maximises} this upper bound between robustness accuracy and accuracy (for any distribution of states) is achieved by the Helstrom measurement corresponding to the optimal discriminator of $\Sigma_A$ and $\Sigma_B$. \\

We can now provide examples of cases where trade-offs between robustness accuracy and accuracy occur and when they do not. \\

\noindent \textbf{Example 5.}
    In the case where the perturbation leaves the states $\Sigma^{\pm}$ unchanged $\tilde{\Sigma}^{\pm}=\Sigma^{\pm}$, then $\tilde{A}_s=A_s$. However, if we have the transformation $\tilde{\Sigma}^{\pm} \rightarrow \Sigma^{\mp}$, then we have \\
\begin{align}
    |\tilde{A}_s-A_s|=|1-2A_s|.
\end{align}
So for $A_s>1/2$, this clearly implies a trade-off between $\tilde{A}_s$ and $A_s$ as $A_s$ grows. We note that even a small perturbation is sufficient to achieve this outcome if the states $\Sigma^{\pm}$ are near the true decision boundaries and while such a modification could be considered an adversarial example because of the small size of the perturbation. We will see this in the next example. For now, we can test this relation on two simple examples of noise: Pauli noise channel and depolarisation noise. \\

\noindent \textit{(Pauli noise)} Suppose we have the Pauli noise channel $\mathcal{N}_{Pauli}$ for a single qubit $\sigma_i$ where 
\begin{align}
    \mathcal{N}_{Pauli}(\sigma_i)=p_0 \sigma_i+p_1 \sigma_i+p_2\sigma_i+p_3 \sigma_i
\end{align}
and $p_0+p_1+p_2+p_3=1$. Suppose we choose the simplest case where $\Pi_0=|0\rangle \langle 0|$. Then we can derive
\begin{align}
    \tilde{\Sigma}^{\pm}=p_0 \Sigma^{\pm}+p_1 \sigma_x \Sigma^{\pm}\sigma_x+p_2 \sigma_y \Sigma^{\pm}\sigma_y+p_3\sigma_z \Sigma^{\pm}\sigma_z.
\end{align}
It is straightforward to find
\begin{align}
    \tilde{A}_s=(1-2(p_1+p_2))A_s+\frac{3}{2}(p_1+p_2).
\end{align}
Then we have a trade-off if $d\tilde{A}_s/dA_s<0$, but not if $\tilde{A}_s/dA_s\geq 0$. So we see that the trade-off condition is true when $p_1+p_2>1/2$, but otherwise there is no trade-off. As special examples, for the case of bit-flip error (where only $p_0, p_1 \neq 0$), the trade-off condition corresponds to $p_2>1/2$, whereas in the case of phase-flip (where $p_1. p_2=0$) we don't have any trade-off between $\tilde{A}_s$ and $A_s$. \\

\noindent \textit{(Depolarisation noise)} Another example is depolarisation noise $\mathcal{N}_{dep}$ where 
\begin{align}
\mathcal{N}_{dep}(\sigma_i)=(1-p)\sigma_i+\frac{p}{d}\mathbf{1}
\end{align}
and
\begin{align}
   \mathcal{N}_{dep}(\Sigma^{\pm})=(1-p)\Sigma^{\pm}+\frac{p}{d}\mathbf{1}.
\end{align}
This leads to 
\begin{align}
    \tilde{A}_s=(1-p)A_s+\frac{p}{2}
\end{align}
where $d\tilde{A}_s/dA_s=1-p>0$, so there will never be a trade-off between robustness accuracy and accuracy for depolarisation noise. \\

\noindent \textbf{Example 6.}
    We call quantum channels $\mathcal{N}_{adv}$ acting on input quantum states $\sigma$ \textit{adversarial perturbations} if the trace distance is below some given constant $\epsilon_{max} \ll 1$
\begin{align}
    \tau(\sigma, \mathcal{N}_{adv}(\sigma)) \leq \epsilon_{max} 
\end{align}
for every input state $\sigma \sim \mathcal{D}$. Then it is possible to have a maximal trade-off between accuracy and  robustness accuracy $\tilde{A}_s=1-A_s$ for adversarial perturbations of the input states if the following sufficient condition is satisfied
\begin{align}
    \min_i \vert \frac{p_{-,i}}{p_{+,i}}\text{Tr}(\sigma_i^+ \sigma_i^-)\vert \geq 1-\epsilon_{max}
\end{align}
where $p_{\pm, i} \equiv P(\sigma^{\pm}_i)/\sum_{\sigma_j^{\pm}}P(\sigma^{\pm}_j)$. This is simple to show using $\mathcal{N}_{Adv}(\Sigma^{\pm})=\Sigma^{\mp}$, and more specifically $\sigma_i^{\pm} \rightarrow \sigma_i^{\mp}$ with normalisation constants taken into account. Then using $\tau(\sigma,\rho) \leq \sqrt{1-\mathcal{F}(\sigma, \rho)}$, the result follows straightforwardly. However, we note that these perturbations are all due to \textit{relevant perturbations} rather than \textit{irrelevant perturbations}, since they change the true class of the states, so they belong to Type II perturbations, rather than Type I. From the above result, we see that the condition of a trade-off depends only on the distances between states of various classes, which is not a property of any particular model in itself, but rather an intrinsic feature of the problem.\\

We can also look at more general classes of trade-offs between robustness accuracy and accuracy. A more general class of examples of trade-offs between $\tilde{A}_s$ and $A_s$ can come about if we consider small perturbations in the distribution $\mathcal{D}$ itself. Let the initial distribution be $\mathcal{D}_1$ and the distribution after perturbation be $\mathcal{D}_2$. We can define these distributions in the following way.\\

For example, suppose we have an observable $\mathcal{O}_j$ and we define $F_j(\sigma_i)=\text{Sign}\text{Tr}(\mathcal{O}_j\sigma_i)$ as the sign of the $j^{th}$ \textit{quantum feature} of the state $\sigma_i$. Then let us define the sign of one quantum feature $F_0$ such that 
\begin{align}
   &  P(F_0(\sigma^{\pm}_i)=\pm 1|c(\sigma_i)=\pm 1)=f_0 \\
    \nonumber 
    & P(F_0(\sigma^{\mp}_i)=\pm 1|c(\sigma_i)=\mp 1)=1-f_0.
\end{align}
So if $f_0>1/2$, then the feature $F_0$ is positively correlated with the true label $c$. We define $\sigma^{f_0^{\pm}}=\{\sigma| f_0(\sigma)=\pm 1\}$ and the following normalised quantum states 
\begin{align}
    & \Sigma^{++}=\frac{\sum_{\sigma_i^+} P(\sigma_i^+)\sigma^{f_0^+}_i}{\sum_{\sigma_j^+} P(\sigma_j^+)} \nonumber \\
    &\Sigma^{+-}=\frac{\sum_{\sigma_i^+} P(\sigma_i^+)\sigma^{f_0^-}_i}{\sum_{\sigma_j^+} P(\sigma_j^+)} \nonumber \\
    &\Sigma^{-+}=\frac{\sum_{\sigma_i^-} P(\sigma_i^-)\sigma^{f_0^+}_i}{\sum_{\sigma_j^-} P(\sigma_j^-)} \nonumber \\
     &\Sigma^{--}=\frac{\sum_{\sigma_i^-} P(\sigma_i^-)\sigma^{f_0^-}_i}{\sum_{\sigma_j^-} P(\sigma_j^-)} \nonumber \\
\end{align}
Then we have the following general result on when to expect a trade-off between robustness accuracy and accuracy. \\

\noindent \textbf{Example 7.}
    Now we study a class of perturbations such that 
\begin{align}
    & \tilde{\Sigma}^{++}=\Sigma^{+-} \nonumber \\
    & \tilde{\Sigma}^{+-}=\Sigma^{++} \nonumber \\
    & \tilde{\Sigma}^{--}=\Sigma^{-+} \nonumber \\
    &\tilde{\Sigma}^{-+}=\Sigma^{--}.
\end{align}
Then it can be shown that for any $f_0>1/2$ we have the trade-off condition
\begin{align}
    \tilde{A}_s \leq \frac{f_0}{1-f_0}(1-A_s).
\end{align}
From the definitions above, these perturbations occur if for instance we exchange the following
\begin{align}
    \sigma^{f_0^+}_i \leftrightarrow \sigma^{f_0^-}_{i'}
\end{align}
so long as the pair $c_i=c_{i'}$. These are irrelevant perturbations. This means that we can always expect a trade-off between robustness accuracy and accuracy if the perturbations move the state labels according to a quantum feature that is strongly correlated with the true label. This analysis is a generalisation of example 5 where we simply exchange
\begin{align}
      \sigma^{+}_i \leftrightarrow \sigma^{-}_{i'}.
\end{align}
Furthermore, following straightforwardly from results in example 6, we see that the above trade-off between robustness accuracy and accuracy for \textit{irrelevant perturbations} can occur if the following sufficient condition is satisfied 
\begin{align}
    \min_i \text{Tr}(\sigma_i^{f_0^+}\sigma_i^{f_0^-}) \geq 1-\epsilon_{max}.
\end{align}
Taking into account non-trivial $f_0$, we can also similarly derive the bound 
\begin{align}
    |\tilde{A}_s-A_s| \leq (2f_0+1)\tau(\bar{\Sigma}_A, \bar{\Sigma}_B)
\end{align}
where 
\begin{align}
    & \bar{\Sigma}_A \equiv 4f_0(\tilde{\Sigma}^{++}+\tilde{\Sigma}^{-+}+\Sigma^{+-}+\Sigma^{--})+2(\tilde{\Sigma}^{+-}+\Sigma^{-+}) \nonumber \\
    & \bar{\Sigma}_B \equiv 4f_0(\tilde{\Sigma}^{+-}+\tilde{\Sigma}^{--}+\Sigma^{-+}+\Sigma^{++})+2(\tilde{\Sigma}^{-+}+\Sigma^{+-}).
\end{align}

\subsubsection{Absence of trade-offs}
We know from Theorem 1 that $A^*=1$ is a \textit{sufficient condition} for there to be \textit{no trade-off} between $\tilde{A}$ and $A$. For example, see \cite{yang}. We can see that this can be a condition on both the input states themselves and a condition on the model $h$. Note that we only wish for robustness of our model against \textit{irrelevant perturbations}, and we wish for non-robustness against \textit{relevant perturbations}. We suppress indices to distinguish these types of perturbations since all robustness is defined with respect to irrelevant perturbations unless otherwise stated. \\

To find the circumstances when the model $H$ under a noisy perturbation gives $A^*=1$, we assume that for all $\sigma \sim D$, we are given 
$\min_{\sigma \sim D} p_{\sigma}$, where $p_{\sigma}=H_k(\sigma)$ (i.e. the probability that $\sigma$ is assigned to the class that it's in) and $\sigma$ is in class $k$. Therefore $\min_{\sigma \sim D} p_{\sigma}$ is a measure of the closest `distance' of the dataset to the decision boundary for the model. Our purpose is to find a sufficient and necessary condition on the upper bound of $\tau(\sigma, \rho)$ (where $H(\rho)=\tilde{H}(\sigma)$), so that we have prediction-change robustness, i.e. $H(\sigma)=H(\rho)=\tilde{H}(\sigma)$ for all states $\sigma \sim D$, which implies $A^*=1$. \\

Locally, i.e. for a specific $\sigma$ where $H(\sigma)\vert_{max} \equiv p_{\sigma}$, we know that a necessary and sufficient condition for prediction-change robustness (i.e. $\tilde{H}(\sigma)=H(\rho)=H(\sigma)$) is 
\begin{align}
    \tau(\sigma, \rho)<\delta_{\sigma}(1-\sqrt{1-\delta_{\sigma}^2})\equiv \tau_{\sigma}
\end{align}
where $\delta_{\sigma}=\sqrt{1-4p_{\sigma}(1-p_{\sigma})}/2$. This is a statement about local prediction-change robustness. However, we care about global prediction change robustness to find the conditions under which $A^*=1$. To change the above into a global statement for all $\sigma \sim D$, this means we need some upper bound of $\tau(\sigma, \rho)$ for all $\sigma \sim D$. We know that a smaller $p_{\sigma}$ should give rise to a larger upper bound to $\tau(\sigma, \rho)$, so we can define the state that gives this smallest probability
\begin{align}
    \sigma_{min} \equiv \text{argmin}_{\sigma \sim D} p_{\sigma}.
\end{align}
Therefore, if we are given no information at all about which states in $\sigma \sim D$ are chosen, then a necessary and sufficient condition for $A^*=1$ is for the perturbation $\tau$ to be below the size
\begin{align}
    \tau<\delta_{\sigma_{min}}(1-\sqrt{1-\delta_{\sigma_{min}}^2}) \equiv \tau_{\sigma_{\min}}. 
\end{align}
Later we can compute conditions for other values $0\leq A^*<1$. \\

We can state this as our Proposition 1, which is a condition on the dataset for some given $H$.\\

\begin{theorem} For the dataset $\sigma \sim D$ and model $H$ subject to noise where $\tilde{H}(\sigma)=H(\rho)$, then if the trace distance between $\sigma$ and $\rho$ satisfies the upper bound 
\begin{align}
     \tau (\sigma, \rho)<\delta_{\sigma_{min}}(1-\sqrt{1-\delta_{\sigma_{min}}^2}),
\end{align}
this is a sufficient condition for $A^*=1$. 
\end{theorem}

We can illustrate this in the following example. We can look at a concrete example of a quantum clustering model where 
\begin{align}
    H_k(\sigma)=1-\tau(\sigma, C_k)
\end{align}
where $C_k$ is the centroid state of the cluster $K$ of states for $\sigma \sim D$. For instance, in the case where $\sigma$ and all $C_k$ are pure states, we have the simple model
\begin{align}
    H_k(\sigma)=F(\sigma, C_k)=\text{Tr}(|C_k\rangle \langle C_k| \sigma),
\end{align}
where the classifier is
\begin{align}
    h(\sigma)=\text{argmax}_k H_k(\sigma).
\end{align}
Clearly, the largest $H_k(\sigma)$ occurs for the state $C_k$ that $\sigma$ is closest to, hence $\sigma$ is identified to belong to cluster $K$. In this case, it is easy to see how the distribution of the sizes of the clusters are sufficient for us to find $p_A$, since 
\begin{align}
    p_A(\sigma)=1-\tau(\sigma, C_A).
\end{align}
Therefore, if we are given the size of the largest cluster, then this gives us $p_{\sigma_{\min}}$, and hence $\tau_{\sigma_{\min}}$. Once inserted into Theorem 3, this would give us the sufficient condition for $A^*=1$ for this model. \\

Aside from $A^*=1$, we see from Theorem 1 that $1>A^*>1/2$ is also sufficient for no trade-offs between robustness accuracy and accuracy. Sufficient conditions for this to be satisfied would also be helpful. From Theorem 3, we saw that for any case where $\tau<\tau_{\sigma_{\min}}$, then $A^*=1$ since $h(\rho)=h(\sigma)$ for all $\sigma \sim D$. However, it may be that we require only $h(\rho)=h(\sigma)$ for some percentage of $\sigma \sim D$ instead of all cases. Suppose we fix $A^*=a$ where $0 \leq a<1$ is some constant. For the distribution $D$, suppose we are given also the distribution of $p_{\sigma}$ values (i.e. the maximum probability of $\sigma$ being in a particular class, which will also be the class assigned to $\sigma$). Let this distribution be labelled $P_{\sigma}(p_A)$, distributed over all $p_A$ values for all $\sigma \sim D$. Then we have the following related results. \\

A sufficient condition for the probability $A^*=a$ to hold is when perturbations of any $\sigma \in D$ satisfies 
\begin{align}
    \tau<\delta_{p_a}(1-\sqrt{1-\delta^2_{p_a}})
\end{align}
where 
\begin{align}
    a=A^*=\int^1_{p_a} P_{\sigma}(p_A) dp_A.
\end{align}
This means that either we can be given some upper bound to $\tau$ where we calculate the corresponding $p_a$ and hence the corresponding $a$ value, or that we are given $a$ and identify the corresponding upper bound to $\tau$. Either way, this requires access to the distribution $P_{\sigma}(p_A)$, which one may or not be able to access. 
\subsubsection{Applications in studying trade-offs}
There are at least three main areas where the study of trade-offs between robustness accuracy and accuracy may be helpful for
\begin{enumerate}
    \item Incompatible noise
    \item Adversarial perturbations
    \item No free lunch theorem
\end{enumerate}
and we will briefly discuss these and present open questions. \\

\noindent \textbf{\textit{Incompatible noise --}}\\

\noindent We first present a definition of \textit{incompatible noise}. 
\begin{definition}
        Let $\tilde{A}^{(1)}_s$ be the robustness accuracy with respect to noise $\mathcal{N}_1$ and $\tilde{A}^{(2)}_s$ is the robustness accuracy with respect to noise $\mathcal{N}_2$. We call two noise perturbations  $\mathcal{N}_1$ and $\mathcal{N}_2$ \textit{incompatible noise} if 
\begin{align}
    \frac{d \tilde{A}^{(2)}_s}{d\tilde{A}^{(1)}_s}<0.
\end{align}
This means that if we train a model to be more accurate in presence of noise $\mathcal{N}_1$, it becomes less accurate when in presence of noise $\mathcal{N}_2$. 
\end{definition}
We can show a simple example of this for a quantum classifier. \\

\noindent \textbf{Example 8.}
As an illustrative simple example for single-qubits with $\Pi_1=|0\rangle \langle 0|$, we let $\mathcal{N}_1$ be depolarisation noise with noise parameter $p_1$ and $\mathcal{N}_2$ be bit-flip noise with noise parameter $p_2$, so $\mathcal{N}_2(\Sigma^{\pm})=(1-p_2)\Sigma^{\pm}+p_2 \sigma_x \Sigma^{\pm} \sigma_x$. Then we can derive
\begin{align}
    & \tilde{A}^{(1)}_s=(1-p_1)A_s+\frac{p_1}{2}\nonumber \\
    &\tilde{A}^{(2)}_s=(1/2-p_2)A_s+(1/2+p_2)/2
\end{align}
hence
\begin{align}
     \frac{d \tilde{A}^{(2)}_s}{d\tilde{A}^{(1)}_s}= \frac{d \tilde{A}^{(2)}_s}{d A_s} \frac{d A_s}{d\tilde{A}^{(1)}_s}=\frac{1/2-p_2}{1-p_1}
\end{align}
In this case, since $p_1<1$, depolarisation noise and bit-flip noise are \textit{incompatible} if and only if $p_2>1/2$.\\

By understanding which perturbations are incompatible with each other, we can better understand the ultimate limitations of our model and to enable better model building. \\

\noindent \textbf{\textit{Adversarial examples --}}\\

\noindent From Example 6, we saw that one way to consider adversarial perturbations are quantum channels $\mathcal{N}_{adv}$ acting on input quantum states $\sigma$ such that the trace distance is below some given constant $\epsilon_{max} \ll 1$
\begin{align}
    \tau(\sigma, \mathcal{N}_{adv}(\sigma)) \leq \epsilon_{max} 
\end{align}
for every input state $\sigma \sim \mathcal{D}$: these include both Type I and Type II perturbations. However, to get true \textit{adversarial examples for state $\sigma$}, we specifically require three conditions:
\begin{itemize}
    \item $\tau(\sigma, \mathcal{N}_{adv}(\sigma)) \leq \epsilon_{max}$, \qquad $\epsilon \ll 1$
    \item $\mathcal{N}_{adv}$ acting on $\sigma$ are also \textit{irrelevant perturbations}, i.e., do not change the true label of $\sigma$: $c(\sigma)=c(\mathcal{N}_{adv}(\sigma))$
    \item $h(\sigma) \neq h(\mathcal{N}_{adv}(\sigma))$
\end{itemize}  
Examples 1-6 all present cases of trade-offs between robustness accuracy and accuracy for \textit{relevant perturbations}. However, in Example 7, we saw that we could also have these trade-offs for \textit{irrelevant perturbations}, and Example 8 also allows the possibility for irrelevant perturbations. Thus trade-offs might also exist in the case of training to make models more robust to adversarial examples. For instance, in the above three conditions, the last deals with the corrupted-instance robustness of the model $A^*$, and we already know from Theorem 1 that a low robustness $A^*<1/2$ to $\mathcal{N}_{adv}$ of Type I perturbations on $\sigma$ is sufficient to guarantee that increasing accuracy $A$ alone will only lead to worse and worse robustness accuracy. We will see examples of this type of behaviour in the next subsection with linear loss. \\

\noindent \textbf{\textit{No free lunch theorem --}}\\

\noindent General trade-off behaviour between robustness acccuracy and accuracy, to both relevant and irrelevant perturbations, can also be considered from the perspective of the no free lunch theorem and also provide ways of constructing the best models. The no free lunch theorem is the observation that for every distribution $\mathcal{D}$ of states for which a model $h$ performs well (e.g. has high accuracy), there should be some distribution $\tilde{\mathcal{D}}$ of states for which the same model $h$ would perform badly. However, the no free lunch theorem doesn't provide an explicit construction of these distributions for any model, since it's only an existence theorem. We would like to understand what the relationship between one distribution where the particular model $h$ works very well, versus a distribution for which the model has low accuracy. 

We can connect the setting of the no free lunch theorem to the relationship between accuracy and robustness accuracy in the following way. \\

We first observe that when we make a perturbation $\mathcal{N}$ on the input state $\sigma$ for the original distribution $\mathcal{D}_1$, we note that we can instead consider this as the same scenario as when selecting $\sigma$ from a modified distribution $\mathcal{D}_2$. This means we can rewrite the loss function for the perturbed states $\tilde{\sigma}=\mathcal{N}(\sigma)$ where $\sigma \sim \mathcal{D}$ as identical to the loss function of the unperturbed states selected from a different distribution $\sigma \sim \mathcal{D}_2$
\begin{align}
   &\tilde{L}_{\mathcal{D}_1}=\textbf{E}_{(\sigma,c)\sim \mathcal{D}_1} l_h(\tilde{\sigma}, c_{\sigma})=\textbf{E}_{(\sigma, c_{\sigma})\sim \mathcal{D}_2} l_h(\sigma, c_{\sigma})=L_{\mathcal{D}_2}.
\end{align}
This simple rewriting leads to the following interpretation. In the context of $1-0$ loss, the first quantity is the corrupted-instance robustness accuracy with respect to $\mathcal{N}$ when selecting states from $\mathcal{D}_1$. The second quantity is accuracy, when selecting states from a different distribution $\mathcal{D}_2$, where $\mathcal{D}_2$ is simply the selection of states in $\mathcal{D}_1$ but all perturbed by $\mathcal{N}$. This means that it's possible to consider the robustness accuracy of one model with respect to a distribution $\mathcal{D}_1$ as identical to the accuracy of the same model with respect to another distribution $\mathcal{D}_2$.\\

Thus we can write $\tilde{A}_1=A_2$. If $A^*$ is the fraction of states $\sigma$ that satisfy $h(\sigma)=h(\mathcal{N}(\sigma))$, then from Theorem 1 we have 
\begin{align}
    A_2=A_1(2A^*-1)+(1-A^*).
\end{align}
Therefore a trade-off relationship between $A_1$ and $A_2$ denotes that the better a model $h$ is with respect to one distribution $\mathcal{D}_1$, the worse the model $h$ is with respect to distribution $\mathcal{D}_2$. We can find the relationship between $\mathcal{D}_1$ and $\mathcal{D}_2$ precisely when we identify the perturbation $\mathcal{N}$ such that $A^*<1/2$. As a worst-case example, when $A^*=0$
\begin{align}
    A_2=1-A_1,
\end{align}
and this means that for the set of states $\sigma \sim \mathcal{D}_1$ and the model $h$ where $h(\sigma)\neq h(\mathcal{N}(\sigma))$ for all states $\sigma \sim \mathcal{D}_1$, when $h$ is the most accurate for $\sigma \sim \mathcal{D}_1$, the same model $h$ is least accurate when the input states are selected from $\mathcal{N}(\sigma)$ where $\sigma \sim \mathcal{D}_1$. \\

Thus the trade-off relationship between robustness accuracy and accuray could also provide a window into more constructive ways of examining the no free lunch theorem and to be able to better pinpoint which models work more effectively for which distributions. 
\subsection{Linear loss}
We can similarly consider models for classification based on linear loss instead of $1-0$ loss, and examine the relationship between the global loss $L$ and the corrupted-instance robust loss $\tilde{L}$. Note that in the previous section for the $1-0$ loss function, the global loss is termed accuracy $A$ and the corrupted-instance robust loss is termed the corrupted instance robustness accuracy $\tilde{A}$. \\

For the linear loss model, we have
\begin{align}
    l_h=-c_i H(\sigma_i)
\end{align}
where $h(\sigma_i)=\text{sign}H(\sigma_i)$. Suppose we have $M$ training states $\{(\sigma_i, c_i)\}_{i=1}^M$. Then we can write
\begin{align}
    L_{\theta}=-\frac{1}{M}\sum_{i=1}^M c_i H_{\theta}(\sigma_i) \nonumber \\
    \tilde{L}_{\theta}=-\frac{1}{M}\sum_{i=1}^M c_i H_{\theta}(\tilde{\sigma}_i)
\end{align}
The relationship between $L$ and $\tilde{L}$ can similarly be applied to identify \textit{incompatible noise}, adversarial perturbations and example, and the no free lunch theorem. We will only provide a few examples here for illustration, since many similar examples already appeared for $1-0$ loss that already motivated some of our fundamental questions. \\

To understand the relationship between $L$ and $\tilde{L}$, we first study an idealised case, where we assume that there are $K_{\theta}=K$ states in the training set such that 
\begin{align} \label{eq:ideal1}
    H(\tilde{\sigma}_i)=H(\sigma_i),
\end{align}
and we denote these as `robust' states $\sigma_{i_R}$. The rest of the $M-K$ states satisfy
\begin{align} \label{eq:ideal2}
    H(\tilde{\sigma}_i)=-H(\sigma_i)
\end{align}
and we denote these states by $\sigma_{i_{\bar{R}}}$. In this case, $K/M=A^*$. Then we can decompose the loss functions in the following way
\begin{align}
    & L=-\frac{1}{M}\left(\sum_{i_{R=1}}^K c_i H(\sigma_{i_R})+\sum_{i_{\bar{R}}=1}^{M-K}c_i H(\sigma_{i_{\bar{R}}})\right) \nonumber \\
    &=\frac{K}{M}L_R+\frac{M-K}{M}L_{\bar{R}}
\end{align}
where $L_R$ and $L_{\bar{R}}$ denote the `robust' and `non-robust' parts of the loss function and is defined as
\begin{align}
    & L_R=-\frac{1}{K}\sum_{i_{R=1}}^K c_i H(\sigma_{i_R}) \\
    & L_{\bar{R}}=-\frac{M-K}{M}\sum_{i_{\bar{R}}=1}^{M-K}c_i H(\sigma_{i_{\bar{R}}}).
\end{align}
Similarly we can rewrite
\begin{align}
   & \tilde{L}=-\frac{1}{M}\left(\sum_{i_{R=1}}^K c_i H(\tilde{\sigma}_{i_R})+\sum_{i_{\bar{R}}=1}^{M-K}c_i H(\tilde{\sigma}_{i_{\bar{R}}})\right) \nonumber \\
   & =\frac{K}{M}L_R-\frac{M-K}{M}L_{\bar{R}}.
\end{align}
Comparing $L$ and $\tilde{L}$, it becomes clear where any possible trade-off between $L$ and $\tilde{L}$ comes from: the non-robust $L_{\bar{R}}$ part of the cost function and a small $A^*$, since
\begin{align}
    \frac{L-\tilde{L}}{2}=\frac{M-K}{M}L_{\bar{R}}.
\end{align}
We observe from the above expression that if the initial good performance of the model -- measured according to $L$ -- is mostly due to the `non-robust' parts $L_{\bar{R}}$, then when it comes to robustness training using $\tilde{L}$, the model will perform poorly. Although $L_{\bar{R}}$ is not quite the prediction-change robust loss $L^*$, this case is still analogous to the case of $1-0$ loss, where if robustness is poor $A^*<1/2$, even if $A$ is very high, $\tilde{A}$ will be low. \\

We can then examine the relationship between $L$ and $\tilde{L}$ for quantum models. 
\begin{theorem}
Suppose we look at quantum models where
\begin{align}
    H(\sigma)=\text{Tr}(\mathcal{M}\sigma).
\end{align}
Then for binary classification with unbiased data we have the difference
\begin{align}
    L-\tilde{L}=\text{Tr}(\mathcal{M}(T^--T^+))
\end{align}
where $T^{\pm} \equiv (\Sigma^{\pm}+\tilde{\Sigma}^{\mp})/2$ and $\sigma_i^{\pm}=\{\sigma_i \vert c_i=\pm 1\}$ and $\Sigma^{\pm}=(2/M)\sum_{i=1}^{M/2}\sigma_i^{\pm}$, $\tilde{\Sigma}^{\pm}=(2/M)\sum_{i=1}^{M/2}\tilde{\sigma}_i^{\pm}$. 
\end{theorem}
\begin{proof}
   This follows easily from the definitions
\begin{align}
    &L=-\frac{1}{M}\left(\sum_{i^+=1}^{M/2}\text{Tr}(\mathcal{M}\sigma_i^+)-\sum_{i^-=1}^{M/2}\text{Tr}(\mathcal{M}\sigma_i^-)\right)=-\frac{1}{2}(\text{Tr}(\mathcal{M}\Sigma^+)-\text{Tr}(\mathcal{M}\Sigma^-)), \\
    & \tilde{L}=-\frac{1}{2}(\text{Tr}(\mathcal{M}\tilde{\Sigma}^+)-\text{Tr}(\mathcal{M}\tilde{\Sigma}^-)).
\end{align}
\end{proof}
Suppose we have two irrelevant perturbations $\mathcal{N}_1$ and $\mathcal{N}_2$. In the idealised case there are $K$ states that obey
\begin{align}
    H(\mathcal{N}_1(\sigma_i))=H(\mathcal{N}_2(\sigma_i))
\end{align}
and $M-K$ states obeying 
\begin{align}
    H(\mathcal{N}_1(\sigma_i))=-H(\mathcal{N}_2(\sigma_i)),
\end{align}
where the former states are `robust' states, whereas the latter are considered `non-robust'. Then the relationship between the $L_1$ and $L_2$ are given by 
\begin{align}
   L_1-L_2=2\frac{M-K}{M}L_{\bar{R}}
\end{align}
where 
\begin{align}
    L_{\bar{R}}=-\frac{1}{M}\sum_{i_{\bar{R}}}^{M-K} c_i H(\sigma_{i_{\bar{R}}})
\end{align}
where $\{\sigma_{i_{\bar{R}}}\}$ are the set of states of the non-robust states. In the more realistic quantum case, we have 
\begin{align}
    L_1-L_2=\text{Tr}(\mathcal{M}(N^--N^+))
\end{align}
where $N^{\pm}\equiv (\Omega_1^{\pm}+\Omega_2^{\mp})/2$ with $\Omega_1^{\pm}=(2/M)\sum_{i=1}^{M/2}\mathcal{N}_1(\sigma_i^{\pm})$ and $\Omega_2^{\pm}=(2/M)\sum_{i=1}^{M/2}\mathcal{N}_2(\sigma_i^{\pm})$. \\

We see that it's also possible to have trade-offs between robust loss and loss for those adversarial perturbations that are in addition irrelevant perturbations. 

\noindent \textbf{Example 9.}
    Here we have an explicit example for a single-qubit model where an irrelevant perturbation is gives rise to a trade-off between the loss and the perturbed loss. For instance, let us have the input states from the training set:
\begin{align}
    &\sigma^+_1=|0\rangle \langle 0| \nonumber \\
    &\sigma^+_2=|1\rangle \langle 1| \nonumber \\
    & \sigma^-_1=\frac{5}{6}|0\rangle \langle 0|+\frac{1}{6}|1\rangle \langle 1| \nonumber \\
    &\sigma^-_2=\frac{1}{3}|0\rangle \langle 0|+\frac{2}{3}|1\rangle \langle 1| \nonumber \\
\end{align}
Suppose under an irrelevant perturbation $\mathcal{N}_{irr}$ we have the following transition
\begin{align}
     &\tilde{\sigma}^+_1=|1\rangle \langle 1| \nonumber \\
    &\tilde{\sigma}^+_2=\frac{1}{4}|0\rangle \langle 0|+\frac{3}{4}|1\rangle \langle 1| \nonumber \\
    & \tilde{\sigma}^-_1=\frac{5}{6}|0\rangle \langle 0|+\frac{1}{6}|1\rangle \langle 1| \nonumber \\
    &\tilde{\sigma}^-_2=\frac{4}{5}|0\rangle \langle 0|+\frac{1}{5}|1\rangle \langle 1| \nonumber \\
\end{align}
In this case we see that 
\begin{align}
    &\Sigma^+=\frac{1}{2}(|0\rangle \langle 0|+|1\rangle \langle 1|) \nonumber \\
    & \Sigma^-=\frac{1}{36}(21|0\rangle \langle 0|+15|1\rangle \langle 1| \nonumber \\
    & \tilde{\Sigma}^+=\frac{1}{8}(|0\rangle \langle 0|+7|1\rangle \langle 1|) \nonumber \\
    & \tilde{\Sigma}^-=\frac{1}{60}(49|0\rangle \langle 0|+11|1\rangle \langle 1|)
\end{align}
Here we see that 
\begin{align}
   & L=\frac{1}{12}\text{Tr}(\mathcal{M}(|0\rangle \langle 0|-|1\rangle \langle 1|)) \nonumber \\
   & \tilde{L}=-\frac{83}{120}\text{Tr}(\mathcal{M}(|0\rangle \langle 0|-|1\rangle \langle 1|)).
\end{align}
This means that $L$ and $\tilde{L}$ cannot be simultaneously minimised. In fact, minimising $L$ leads to $\mathcal{M}=|1\rangle \langle 1|-|0\rangle \langle 0|$, which would in fact maximise $\tilde{L}$. \\

We note that there are fundamental reasons why linear loss with linear quantum models are not always sufficient to contain the symmetries of the true underlying models. For instance, we know (from previous theorems [ ]) that if $\tau(\sigma, \rho)<p_{\sigma}-1/2$, then it is the case that $h(\sigma)=h(\rho)$ for the linear quantum model $h$. However, suppose we have $\sigma=\sigma^+$ and $\rho=\mathcal{N}(\sigma^-)$, which means that $c(\rho)=-c(\sigma)$. If in addition $\tau(\sigma^+, \tilde{\sigma}^-)<p_{\sigma}-1/2$, this means that $h(\rho)=h(\sigma)$, which cannot represent $c$ exactly. Therefore, if enough states $\sigma$ satisfy $\tau(\sigma^+, \tilde{\sigma}^-)<p_{\sigma}-1/2$, then the ability of $h$ to represent $c$ is low. We can also consider the following example \\

\noindent \textbf{Example 10.}
    We have the training states $\{\sigma^+_1, \sigma^+_2, \sigma^-_3, \sigma^-_4\}$ and the states modified by irrelevant perturbations, so that the test states are $\{\tilde{\sigma}^+_1, \tilde{\sigma}^+_2, \tilde{\sigma}^-_3, \tilde{\sigma}^-_4\}$, where by definition 
\begin{align}
    &c(\sigma^+_1)=c(\sigma^+_2)=c(\tilde{\sigma}^+_1)=c(\tilde{\sigma}^+_2)=1 \nonumber \\
    &=-c(\sigma^-_3)=-c(\sigma^-_4)=-c(\tilde{\sigma}^-_3)=-c(\tilde{\sigma}^-_4)
\end{align}
In this example, we see that when $L_{train}$ is optimised there is a large trade-off with $L_{test}$. Suppose we use the linear loss using $h(\sigma)=\text{sign}H(\sigma)$ so that $L=-\sum_i c_i h(\sigma_i)$. Now we see that the condition 
\begin{align}
    \tau(\sigma_i^{\pm}, \mathcal{N}_{irr}({\sigma}^{\mp}_i))\leq p_{\sigma_i}-1/2
\end{align}
for $i=1,2,3,4$ is sufficient to make $L=-\tilde{L}$ when $L$ is optimised. For instance, the above condition is sufficient to ensure that $h$ does \textit{not} preserve the same symmetries as $c$, while still optimising $L$, since
\begin{align}
    c(\sigma^{\pm}_i)=h(\sigma^{\pm}_i)=\pm 1
\end{align}
while 
\begin{align}
    c(\tilde{\sigma}^{\pm}_i)=-h(\tilde{\sigma}^{\pm}_i)=\pm 1.
\end{align}
Then we see that while $L=-4$, we get $\tilde{L}=4$. In this case, we see this to be a consequence of the states $\mathcal{N}_{irr}(\sigma^{\pm})$ being too close to $\sigma^{\mp}$ in trace distance, while they are actually differently classified. \\

Note that this is \textit{not} necessarily saying that $\mathcal{N}_{irr}$ is an adversarial perturbation, since it is possible for $\tau(\sigma_i^{\pm}, \mathcal{N}_{irr}({\sigma}^{\mp}_i))$ to be small, while $\tau(\sigma_i^{\pm}, \mathcal{N}_{irr}({\sigma}^{\pm}_i))$ not to be small (hence not adversarial).\\

From these examples, it is clear that it is necessary in some cases it is better that the symmetries of the true classification be explicitly included. 
\subsection{Role of relevant and irrelevant features}
We know from previous examples that standard training (optimising accuracy) and robustness training (optimising robustness accuracy) for states selected from the \textit{same} distribution can give rise to different outcomes in terms of how well theu ultimately perform, and this is because different properties or \textit{features} of the input states are being learned or prioritised. Although we might have an intuitive feeling what a feature might mean (e.g. the roundness of a ladybug, the cuteness level of a cat, the toxicity of a death-cap mushroom, the purity of a quantum state...etc), we must begin by formally defining what we might mean by a feature. \\

\begin{definition} 
We can define a \textit{feature} $f_{\Lambda}$ with respect to operator $\Lambda$ acting on a quantum state $\sigma$ as $f_{\Lambda}: \sigma \rightarrow \mathbb{R}$, where 
    \begin{align}
        f_{\Lambda}(\sigma)=\text{Tr}(\Lambda(\sigma)).
    \end{align}
    Then we can call a feature robust or non-robust depending on how well-correlated this feature is to the true label $c(\sigma)$ of $\sigma$. In principle $\Lambda$ can also be nonlinear (e.g. $\Lambda(\sigma)=\sigma^2$, as in the case of obtaining the purity of a quantum state).  
    \end{definition}

    \begin{definition} We call a feature $f_{\Lambda}$ \textit{$k$-useful} for a given distribution $\mathcal{D}$ if 
    \begin{align}
        \textbf{E}_{(\sigma, c)\sim  \mathcal{D}}(c(\sigma) \cdot f_{\Lambda}(\sigma))\geq k
    \end{align}
    \end{definition}
    
    \begin{definition} 
    We say $f_{\Lambda}$ is a \textit{robust feature} with respect to $(\mathcal{D}, \mathcal{N}_{\delta}, k, \gamma)$ if it is $k$-useful and that under perturbations $\mathcal{N}_{\delta}$ it remains $\gamma$-useful, so
    \begin{align}
        \textbf{E}_{(\sigma, c)\sim  \mathcal{D}}(\inf_{\delta} c(\sigma) \cdot f_{\Lambda}(\mathcal{N}_{\delta}(\sigma)))\geq \gamma. 
    \end{align}
    \end{definition}
    Then a useful but non-robust feature is a feature which is $k$-useful for some $k>0$ but is not $\gamma$-robust for any $\gamma \geq 0$. Note that the above robust feature definition can also be extended to be defined with respect to different loss functions and having different versions of robust features in the same way we have different robustnesses in Section 2.\\
    
    A working hypothesis is that, given  $\mathcal{D}$ and hypothesis class $\mathcal{H}$, if the (standard) accuracy relies heavily on $k$-useful but non-robust features, then the accuracy can be high while the robustness accuracy remains low. This would identify the cause of these trade-off relationships and provide a means to more efficiently resolve these, if possible. \\
    
    To express this mathematically, we have the following conjecture.\\
    
    \noindent \textbf{Conjecture 1.} Let the feature $f_{\Lambda}$ be $k-$useful, so  let $\textbf{E}_{(\sigma, c)\sim  \mathcal{D}}(c(\sigma) \cdot f_{\Lambda}(\sigma))\geq k_{\Lambda}$ and it is non-robust, so $\textbf{E}_{(\sigma, c)\sim  \mathcal{D}}(\inf_{\delta} c(\sigma) \cdot f_{\Lambda}(\mathcal{N}_{\delta}(\sigma)))\leq \gamma_{\Lambda}$. Then if the hypothesis function $h$ has high standard accuracy so $A>a>1/2$, then it is possible to show that this provides an upper bound to the robustness accuracy $\tilde{A}<g(k_{\Lambda}, \gamma_{\Lambda}, a)$, where larger $k_{\Lambda}$, small $\gamma$ and larger $a$ leads to smaller $\tilde{A}$.\\
    
    There are some models which demonstrate this intuition and we discuss simple cases. \\
    
\noindent \textbf{Example 11.}
     We can give a concrete demonstration of a quantum example, where we have the product state $\sigma=\sigma_0\otimes \sigma_1 \otimes...\otimes \sigma_d$. We can also consider this a map $x_i\rightarrow \sigma_i$. Other possibilities exist, but  we start with a simple example. \\
     
     The distribution we are selecting from $(\sigma, c(\sigma)) \sim \mathcal{D}$ is specified. So we let $\{p_i, \sigma_i\}$ to be the $\sigma$ states we are selecting from, where $p_i$ are the probabilities in which density matrices $\sigma_i$ are sampled. Note that this means the full density matrix describing the distribution $\bar{\sigma}=\sum_i p_i \sigma_i$ is \textit{not} part of the distribution $\mathcal{D}$. \textit{Only} the states $\rho_i$ are included. \\
     
     Suppose each $\sigma_i=|\sigma_i\rangle \langle \sigma_i|$ are single pure qubits and $|\sigma_0\rangle=\cos(\theta_0/2)|0\rangle+\exp(i \phi) \sin(\theta_0/2)|1\rangle$. Let the true class $c(\sigma)$ be \textit{defined} as whether the qubit $\sigma_0$ is in the upper hemisphere of the Bloch sphere or not, so 
     \begin{align}
         c(\sigma)=\text{sgn}(\text{Tr}(|0\rangle \langle 0| \sigma_0)-1/2)=\text{sgn}(\cos(\theta_0^2/2)-1/2).
     \end{align}
      This means if $\sigma_0$ is in the upper hemisphere, or $\theta_0 \in [0, \pi/2)$, then $c(\sigma)=1$ and otherwise $\theta_0 \in [\pi/2, \pi]$ so $c(\sigma)=-1$. \\
      
      We can rewrite $c(\sigma)=\text{sgn}\text{Tr}((\lambda_0 \otimes \mathbf{1}^{\otimes d}) \sigma)$ where $\lambda_0=|0\rangle \langle 0|-(1/2)\mathbf{1}=(1/2)(|0\rangle \langle 0|-|1\rangle \langle 1|)$. However, the classifier $h(\sigma)$ need not be perfect. Suppose it misclassifies a fraction $p$ of states in the upper hemisphere but classifies all the states in the lower hemisphere correctly. Then we can write such a hypothesis as 
      \begin{align}
          h_0(\sigma)=\text{sgn}(\text{Tr}((\lambda'_0 \otimes \mathbf{1}^{\otimes d}) \sigma) 
      \end{align}
      where $\lambda'_0=|0\rangle \langle 0|-(1/2+\eta)\mathbf{1}=(1/2-\eta)|0\rangle \langle 0|-(1/2+\eta)|1\rangle \langle 1|$.\\
      
      Suppose in our distribution half of the states $\sigma_0$ are sampled from the upper hemisphere and half the states are sampled from the lower hemisphere, so $P(c(\sigma)=1)=1/2=P(c(\sigma)=-1)$. By construction of our hypothesis, we define that $P(c(\sigma)=1=-h_0(\sigma))=p$ and $P(c(\sigma)=-1=h_0(\sigma))=1/2$. If we use this hypothesis, then we know that the accuracy is the probability
      \begin{align}   & P(c(\sigma)=h_0(\sigma))=P(c(\sigma)=1=h_0(\sigma))+P(c(\sigma)=-1 \nonumber \\
      &=h_0(\sigma))=(1/2-p)+1/2=1-p.
      \end{align}
     This classifier $h_0(\sigma)$ is quite robust in the sense that the distribution of $\sigma_0$ (originally evenly distributed over the whole Bloch sphere) must be modified substantially before the accuracy decreases (i.e. moving $\sigma_0$ from the top hemisphere all towards the equator around which the fraction of $p$ states are). \\
     
     However, this $h_0(\sigma)$ is not the only hypothesis. Suppose we have observables on the other states $\sigma_1,...,\sigma_d$, but these are more weakly correlated with $c(\sigma)$. For instance, let the distribution of $\sigma_{l \neq 0}$ be such that the feature
     \begin{align}
         f_l(\sigma)=\text{Tr}(\Lambda_l \sigma)=\text{Tr}(\lambda_l \sigma_l)
     \end{align}
     has a probability distribution centered on $\pm|\xi|<<1$ for $c(\sigma)=\pm 1$ and we call these distributions $G_+$ and $G_-$ respectively. The local operators $\lambda_l=(1/2)(|0\rangle \langle 0|-|1\rangle \langle 1|$ and the details of the distributions $G_{\pm}$ comes from the choice in states $\sigma_1,...,\sigma_d$ (e.g., unlike for $\sigma_0$ which are chosen randomly across the entire Bloch sphere, $\sigma_{l \neq 0}$ are chosen to lie only around the equator of the Bloch sphere where the decision boundary is.\\
     
     Suppose for these distributions
     $P_{G_+}(f_l(\sigma)=-1=-c(\sigma))=\bar{p}=P_{G_-}(f_l(\sigma)=1=-c(\sigma))$ where $\bar{p}<1/2$. Since $|\xi|$ is very small, we expect $\bar{p}$ to be quite close to $1/2$, which is quite large. This means a large probability of misclassification. For instance, if we used just the hypothesis function $h_l(\sigma)=\text{sgn}f_l(\sigma)$, then the accuracy is 
     \begin{align}
         P(c(\sigma)=h_l(\sigma))=1-\bar{p}.
     \end{align}
     However, suppose we added all of the $l\neq 0$ features together. So each feature on its own is only weakly correlated with $c(\sigma)$, but together they could make a better classifier in the following way. Let 
     \begin{align}
        h(\sigma) \equiv (1/d) \text{sgn}(\sum_{l=1}^d f_l(\sigma)).
     \end{align}
     Now, as $d$ grows, we know that the standard deviation the sum of the $G_{\pm}$ distributions decreases (central limit theorem). In the limit $d \rightarrow \infty$, the \textit{sum} of the $G_{\pm}$ distributions should be delta functions centered on $\pm|\xi|$. This means that the probability the whole sum should be of the opposite sign goes from $\bar{p}$ (for a single $l$) to a probability $\bar{p}'<<\bar{p}$, which drastically increases the accuracy $P(c(\sigma)=h(\sigma)$, even though $h(\sigma)$ consists only of features that are weakly correlated with the true class. \\
     
     These weakly correlated features are also non-robust features, we can perturb the distributions of $\sigma_1 \otimes...\otimes \sigma_d$ slightly so that $G_{\pm} \leftrightarrow G_{mp}$. One possibility is that one can write
     \begin{align}
        & P(h(\sigma')=c(\sigma))=\frac{1}{2}(P_{G_+}(h(\sigma')=1|c(\sigma)=1)+P_{G_-}(h(\sigma')=-1|c(\sigma)=-1)) \nonumber \\
        &=\frac{1}{2}(P_{G_-}(h(\sigma)=1|c(\sigma)=1)+P_{G_+}(h(\sigma)=-1|c(\sigma)=-1)).
     \end{align}
     However, in this form  in the last line it cannot be compared directly to the original standard accuracy
     \begin{align}
         P(h(\sigma)=c(\sigma))=\frac{1}{2}(P_{G_+}(h(\sigma)=1|c(\sigma)=1)+P_{G_-}(h(\sigma)=-1|c(\sigma)=-1)).
     \end{align}
     One method to compare these terms is to break up $P(h=1|c=1)$ terms into $P(h=1|f_0=1)P(f_0=c)+P(h=1|f_0=-1)P(f_0=-c)$, so just like in the classical example above. The proof then follows in exactly the same way. Then it is clear how the trade-off appears in the quantum case, that one can train classifiers to be very accurate, but they will not be robust. \\

\noindent \textbf{Example 12.}
     We can extend to more general setups, so instead of making the $(x_0,x_1,...,x_d) \rightarrow (\sigma_0, \sigma_1,...,\sigma_d)$ correspondence, we can also map the the amplitudes of a $d$-dimensional state $\sigma$, or some other variant. \\
     
     So we adopt the following ansatz, where $h(\sigma)=\text{sgn}H(\sigma)$
     \begin{align}
         H(\sigma)=\sum_{l}w_l f_l(\sigma)=\sum_l w_l \text{Tr}(\Lambda_l \sigma).
     \end{align}
     We want to generalise the notion of what is a robust and what is a non-robust feature. So in the previous example, we saw clearly how projection of $|0\rangle \langle 0|$ onto $\sigma_0$ gives robust features, but projection onto $\sigma_{l \neq 0}$ does not. Of course, the robustness properties are dependent very much on the distributions from which the states $\sigma$ are selected. The key goal is to identify tests on the individual features to test for both their contribution to accuracy and their contribution to the accuracy when a perturbation is present. The details of the analysis will differ depending the on the loss functions chosen.\\
     
     For example, for the $1-0$ loss function, the interesting quantum quantities are 
     \begin{align}
         P(c(\sigma)=w_l f_l(\sigma))
     \end{align}
     and
     \begin{align}
         P(c(\sigma)=w_l f_l(\sigma')).
     \end{align}
     If the former is high and the latter is even lower, then know this $l^{\text{th}}$ feature is a feature that contributes a lot to the standard accuracy and robustness accuracy trade-offs. \\
     
     For other loss functions, it's some other correlation between $c(\sigma)$ and $w_l f_l(\sigma)$. \\
     
     In the end, what we want is a decomposition of $A-\tilde{A}$ in terms of the difference between the accuracies for every individual feature $f_l$ so we can identify the features most responsible for the trade-off and try to eliminate those non-robust features and to inject more robust features. So the procedure is
     \begin{align}
         \text{Choose loss function} \rightarrow A(l)-\tilde{A}_{\delta}(l) \rightarrow A-\tilde{A}_{\delta},
     \end{align}
     where the subscript $\delta$ denotes that the robustness accuracy is defined with respect to the perturbation $\mathcal{N}_{\delta}$. 
     The first arrow depends on the loss function chosen. However, the last arrow may also be difficult because the total accuracy may not be a linear sum of the individual components.\\
     
     However, we can first examine a case where such a linear relationship exists. So when $H(\sigma)=\sum_l H_l(\sigma)$ and $h_l(\sigma)=\text{sgn} H_l(\sigma)$, let the \textit{linear loss function} be written as
     \begin{align}
        \textbf{E}_{(\sigma, c) \sim D} \mathcal{L}(\sigma, c(\sigma), H(\sigma))=\textbf{E}_{(\sigma,c) \sim D} \sum_{l} \mathcal{L}(\sigma, c(\sigma), H_l(\sigma)).
     \end{align}
     In this case 
     \begin{align}
         A-\tilde{A}_{\delta}=\sum_l (A(l)-\tilde{A}_{\delta}(l))
     \end{align}
     and in this way the dominant $l$ components can be identified. These linear loss functions will then be linear functions of $H(\sigma)$, which is in turn linear in the features $f_l(\sigma)$. \\
     
     For these models, we can separate robust and the most non-robust features in the following way. By labelling robust features by $R$ and non-robust features by $\bar{R}$, we can write
     \begin{align}
         &H_{l \in R} (\mathcal{N}_{\delta}(\sigma))=H_{l \in R} (\sigma) \nonumber \\
         &H_{l \in \bar{R}} (\mathcal{N}_{\delta}(\sigma))=-H_{l \in \bar{R}} (\sigma).
     \end{align}
     Of course, the exact relationships are simplified in the above example for an idealised case, but it can be generalised. Then if $\mathcal{L}(\sigma, c(\sigma), H_l(\sigma))=-c(\sigma) H_l(\sigma)$ we have 
     \begin{align}
    &|A-\tilde{A}_{\delta}|=|\textbf{E}_{(\sigma,c)\sim D} c(\sigma) (\sum_l^L-\sum_{l \in R})H_{l}(\sigma)+\sum_{l \in \bar{R}} H_l(\sigma)| \nonumber \\
    &=2|\textbf{E}_{(\sigma,c)\sim D} c(\sigma) \sum_{l \in \bar{R}} H_l(\sigma)|
     \end{align}
     where we used $\sum_l^L=\sum_{l \in R}+\sum_{l \in \bar{R}}$. From this analysis, we see that the trade-off is due entirely to the non-robust features, and if the non-robust features were to disappear, then the trade-offs would disappear. \\
     
     What about \textit{nonlinear} loss functions in the features? For instance, if we have a nonlinear function $g(\cdot)$ as a function of $-c(\sigma)H(\sigma)$, then
     \begin{align}
         & |A-\tilde{A}_{\delta}|= \nonumber \\
         &|\textbf{E}_{(\sigma,c)\sim D}(g(c(\sigma)(\sum_{l \in R}H_l(\sigma)+\sum_{l \in \bar{R}}H_l(\sigma)))-g(c(\sigma)(\sum_{l \in R}H_l(\sigma)-\sum_{l \in \bar{R}}H_l(\sigma))))|
     \end{align}
     So if $g(\cdot)$ is a $K$-Lipschitz continuous function, then 
     \begin{align}
         |A-\tilde{A}_{\delta}|\leq 2K|\textbf{E}_{(\sigma,c)\sim D} c(\sigma)\sum_{l \in \bar{R}} H_l(\sigma)|.
     \end{align}
     This means that even for a very accurate classifier that relies on non-robust features, the amount of trade-off is dependent at most on the non-robust elements only. \\

     The above analysis points out that the key quantities to figure out are the non-robust elements: what they are and how they are related to the distribution $\mathcal{D}$ and the perturbation $\mathcal{N}_{\delta}$. Note that the really problematic features are the ones that are \textit{both non-robust and useful}, where the latter means they contribute to a high $A$. \\
     
     So the key question is: how to figure out which features are both useful and non-robust, hence one needs to get rid of them in one's classifier? Furthermore, once they are eliminated, how to replace them by robust features so that the accuracy is not compromised? \\
     
     Note that in the above inequality, the right-hand side is related to the usefulness of the non-robust features to the standard accuracy. So if those terms are low, then the difference between the accuracy and robustness accuracy must likewise be low. Otherwise, they can be large. \\

      One can also apply spectral analysis to identify better those features that are both useful and robust, for example see techniques in \cite{yin2019fourier} for computer vision. 
  
\section{Connection between robustness in quantum machine learning and dynamical systems}
Robustness in machine learning (both classical and quantum) is not only a 'static' property of a trained model, as has been considered thus far. In fact, closely related concepts to robustness -- numerous forms of stability, for instance -- emerges from the study of dynamical processes that can govern both the training stage and inference stage in continuous-time analog models. The notion of stability comes in many forms and are heavily studied in dynamical systems, including Lyapunov stability, asymptotic stability, stability in optimisation and stability in signal processing and control theory (e.g., bounded input-bounded output BIBO). While stability is chiefly concerned with how a system reacts to small perturbations in initial conditions or
parameters, often focusing on the system's ability to return to an equilibrium point, different notion of robustness refers to good performance in the presence of uncertainty and perturbations.\\

By framing learning as a dynamical system, we can gain access to powerful mathematical tools — from ordinary and partial differential equations, stochastic calculus, optimal control to quantum mechanics, which can both sharpen and broaden our theoretical perspective to unlock new architectures and training paradigms that can perform better in the presence of uncertainty, noise and adversarial perturbations. We just mention a few possibilities here. \\

Gradient-based methods for training can be considered in terms of dynamical systems. Training a machine learning model using gradient descent can be understood as evolving a parameter vector $\theta(t)$ under the influence of a force field derived from the loss function 
$\mathcal{L}({\theta}(t))$ . In the continuous-time limit, this evolution satisfies the gradient flow differential equation
\begin{align}
    \frac{d \theta(t)}{dt}=-\nabla_{\theta}\mathcal{L}(\theta(t))
\end{align}
which is a deterministic dynamical system. The trajectories of this flow dictate how the model evolves toward local minima. Studying the stability of these flows, particularly through the Lyapunov framework, allows us to assess convergence properties and sensitivity to initialization. The robustness of the resulting model is thus linked to the geometry of the loss landscape and the conditioning of the flow — flatter minima often correspond to more robust solutions. \\

Machine learning algorithms can also be built on continuous-time forms of residual neural networks (both classical and quantum), also called neural ordinary and neural partial differential equations (NODEs and NPDEs). For example the model $h$ evolves with continuous depth $t$ as
\begin{align}
    \frac{dh(t)}{dt}=f(h(t), t; \theta). 
\end{align}
 Stability in dynamical systems refers to the behavior of a system under small perturbations to its internal state or initial conditions, and here the stability of the solutions is governed by the Jacobian $J=\partial f/ \partial h$. If all eigenvalues of the Jacobian have non-positive real parts, the trajectories of $h(t)$ are contractive or non-expansive, implying that small perturbations in initial conditions $h(0)$ do not grow uncontrollably with time. This stability ensures that the system behaves predictably and is essential for numerical integration and convergence. \\
 
 In these dynamical systems, the concepts of stability and robustness, while closely related, play distinct roles and are governed by different mathematical principles. Robustness, on the other hand, is concerned with the model’s output sensitivity to perturbations in its input, including noise, distribution shifts, or adversarial attacks. While a stable NODE might suppress the effects of internal perturbations, robustness assesses the model’s functional generalisation—its ability to maintain output performance despite changes or uncertainty in the input state $\sigma$. If in addition $\|f(h_1, t)-f(h_2,t)\| \leq L\|h_1-h_2\|$ $\|f(\mathcal{N}_{\delta}(\sigma))-f(\sigma)\|$, the system is Lipschitz continuous, and the solution is stable under bounded perturbations. This property translates to robustness against adversarial perturbations, since the model resists sharp changes in output due to small input noise. \\

In summary, while stability is a mathematical property of the internal learning dynamics, and often a prerequisite for training convergence and numerical solvability, robustness relates to the reliability of a model’s predictions under external disturbances. The two are connected: stability in Neural ODEs and PDEs often enhances robustness, but the reverse is not guaranteed. A model can be stable but still vulnerable to specific structured input attacks. Understanding both are critical in designing machine learning systems that are dependable in real-world, noisy, or adversarial settings.\\

On a related note, robustness can also be framed through feedback control, where training adapts to perturbations using corrective signals. Consider the controlled learning dynamics
\begin{align}
    \frac{d \theta(t)}{dt}=-\nabla_{\theta} \mathcal{L}(\theta(t))+K(\sigma, \delta)
\end{align}
where $K$ is a feedback term based on perturbation $\mathcal{N}_{\delta}$ with 'size' $\delta$ - however one chooses to define the size according to a specific metric. By designing 
$K$ using robust control techniques (e.g., H-infinity control), one can stabilise the learning process even under worst-case data corruptions. This control-theoretic lens provides a new way to consider for adversarial robustness and other types of robustness from the viewpoint of dynamical systems, for both classical and quantum learning.\\

The language of stability, more familiar in the study of ODEs and PDEs, can be relevant in other ways to notions of robustness of both classical and quantum machine learning. For example, adversarial training solves a minimax optimization problem
\begin{align}
    \min_{\theta} \mathbf{E}_{(\sigma,c)} \max_{\delta} \mathcal{L}(\mathcal{N}_{\delta}(\sigma), c, \theta).
\end{align}
This can be interpreted through the Hamilton-Jacobi-Bellman (HJB) equation, which governs optimal control under adversarial dynamics. In this formulation, the adversary acts as a controller that perturbs the input, and training seeks a policy (parameter $\theta$) that minimises the worst-case loss. Viewing adversarial training through this lens facilitates PDE-based approximations that lead to robust classifier architectures. Here the differences between classical and quantum cases would lie in different measures of distance, different structure of the loss function, different classes of 'natural' input states....etc. \\

There are also machine learning algorithms that are designed for the learning of different dynamical systems, e.g. physically-informed neural networks (PINNs), and in these cases it's clear that there is a potential rich interplay of the various notions of stability and robustness for these applications. 
\section{Discussion and open questions}
From the theorems and and examples we've explored in this chapter, we can see that it is very important to be careful of the definitions that are used, for example robustness, accuracy and robustness accuracy. Even from simple examples, these suggest some preliminary conclusions:
\begin{itemize}
    \item Accuracy and robustness accuracy trade-offs can occur due to relevant perturbations. 
    \item Accuracy and robustness accuracy trade-offs can also occur due to irrelevant perturbations, for example, when the model has low accuracy (below a certain threshold) as well as a high model robustness. Alternatively, it can also occur when the model has high enough accuracy (above a threshold), but the model has low model robustness. This can be due to the hypothesis class for the model not allowing for the symmetry corresponding to the irrelevant perturbation. Here a higher accuracy gives a lower robust accuracy. This in turn shows that robustness accuracy can be improved with higher model robustness against irrelevant perturbations if accuracy is already high enough.
    \item It is possible to find sufficient conditions where accuracy and robusntess accuracy trade-offs do not occur at all.
    \item When training states include the the same symmetry corresponding to the irrelevant perturbations, the trade-off condition can be shown to improve and even disappear altogether under fairly general conditions.
     \item To distinguish the effect of  relevant and irrelevant perturbations, one can more effectively improve the models by  explicitly including known symmetries.
    \item  It is possible that when a model is trained under \textit{incompatible noise}, better accuracy under training with one type of perturbation can induce worse accuracy when training with another type of perturbation. This is particularly important when requiring robustness of quantum models, to explore especially when one trains models under quantum noise
    \item The robustness accuracy and accuracy trade-off relationship can inform a constructive form of the no free lunch theorem, which can help suggest better inductive biases for the distributions under which the model works well
\end{itemize}
We see that even simple examples in this chapter already suggests many open questions, for example  
\begin{enumerate}
\item	\textbf{General Trade-off Characterization:} Can a complete taxonomy of trade-off types be formulated, especially in higher-dimensional quantum systems with non-trivial entanglement and noise correlations? What are the chief differences between classical and quantum examples? 
\item \textbf{Incompatible noise:} Can a taxonomy be developed more systematically to identify incompatible quantum noise, and what are the implications for the training of quantum models? 
\item \textbf{No free lunch theorem:} How can better understanding of trade-off relationships between different forms of robustness and accuracy provide a more constructive form of the no free lunch theorem, to suggest better inductive biases? What are the differences between classical and quantum cases? 
\item	\textbf{Robust Training Strategies:} What are the optimal training algorithms and cost functions that prioritize robustness without excessively compromising accuracy? Can quantum-specific regularization strategies be developed? 
\item	\textbf{Robustness in Regression and Unsupervised Learning:} This chapter focused on classification problems. Can these definitions and trade-off principles be extended rigorously to regression tasks or quantum unsupervised learning frameworks?
\item	\textbf{Adversarial Quantum Examples:} Is it always possible to construct quantum adversarial examples efficiently where no trade-offs between accuracy and robustness accuracy exists, and are there practical defenses against them in near-term quantum hardware?
\item	\textbf{Model Selection and Hypothesis Class Constraints:} Are there universal bounds on robustness accuracy for specific classes of quantum models?
\item	\textbf{Connections to Dynamical Systems:} We noted parallels with dynamical stability and control theory. Can ideas from Lyapunov stability or Hamiltonian systems yield new insights into robustness in quantum learning, in a more system way?
\end{enumerate}

\section*{Acknowledgements}
This book chapter (long overdue) is dedicated to the memory of a good colleague and friend Peter Wittek, with whom the author wrote the first quantum paper studying the robustness of quantum classification algorithms against adversarial perturbations.  \\

 NL acknowledges funding from the Science and Technology Commission of Shanghai Municipality (STCSM) grant no.24LZ1401200 (21JC1402900) and the Shanghai Pilot Program for Basic Research. NL is also supported by NSFC grants No.12471411 and No. 12341104, the Shanghai Jiao Tong University 2030 Initiative, the Shanghai Science and Technology Innovation Action Plan (24LZ1401200) and the Fundamental Research Funds for the Central Universities.  

\section*{Appendix A: Proof of Theorem 1} \label{app:A}
Let $C_{\sigma} $ be the class label output of our noiseless model for input state $\sigma$.  However,  this model does not necessarily coincide with the `ground truth' for $\sigma$,  which we present by $T_{\sigma}$.  Thus the model is correct for $\sigma$ if $T_{\sigma}=C_{\sigma}$ and is incorrect for $\sigma$ if $C_{\sigma} \neq T_{\sigma}$.  For both our training and our test states,  let us sample from the states $\sigma$ from some (usually unknown) distribution $\mathcal{D}$.  Then by \textit{accuracy in the absence of noise} we mean the probability that $T_{\sigma}=C_{\sigma}$ when $\sigma \sim \mathcal{D}$,  which can be denoted 
\begin{align}
A \equiv P_{\sigma \sim \mathcal{D}}(C_{\sigma}=T_{\sigma}) \equiv P(C=T)
\end{align} 
where we have dropped the subscripts for convenience.  Now let us include noise in our model so $C_{\sigma} \rightarrow \tilde{C}_{\sigma}$. Then the \textit{accuracy in the presence of noise} can denoted by the probability
\begin{align}
\tilde{A} \equiv P_{\sigma \sim \mathcal{D}}(\tilde{C}_{\sigma}=T_{\sigma}) \equiv P(\tilde{C}=T)
\end{align}
There is \textit{another} type of accuracy which we can call \textit{robustness accuracy},  which refers to the probability that the model itself gives rise to the same prediction after adding noise,  irrespective of the relationship to the ground truth.  This robustness accuracy we can denote by 
\begin{align}
A^*=P_{\sigma \sim \mathcal{D}}(\tilde{C}_{\sigma}=C_{\sigma}) \equiv P(\tilde{C}=C).
\end{align}
Usually it is $A$ and $\tilde{A}$ that is of interest for generalisation performance,  so by relating $\tilde{A}$ with $A$,  we can determine by what amount accuracy degrades in the presence of noise.  However,  we will see that this relationship will also depend on $A^*$ and it is through $A^*$ that we can include information about the type of noise that is added.  \\

Let there be $K$ classes,  so each of $C, \tilde{C}, T$ can take values $0, 1,...,K-1$.  So now we can rewrite $\tilde{A}$ as 
\begin{align} \label{eq:tildeAjj}
& \tilde{A} \equiv P(\tilde{C}=T)=\sum_{j=0}^{K-1}  P(\tilde{C}=T=j)=\sum_j P(\tilde{C}=j|T=j)P(T=j)=\frac{1}{K} \sum_j P(\tilde{C}=j|T=j) \nonumber \\
&=\frac{1}{K}\left(\sum_{j} P(\tilde{C}=j|C=j)P(C=T)+\sum_j\sum_{k \neq j} P(\tilde{C}=j|C=k\neq j)P(C \neq T)\right) \nonumber \\
&=\frac{1}{K}\left(A\sum_{j} P_{jj}+(1-A) \sum_j \sum_{k \neq j} P_{jk} \right) 
\end{align}
where $P_{ij} \equiv P(\tilde{C}=i|C=j)$.  We also used the generic assumption in the first line that the samples we have are `unbiased', in the sense that there are as many states of one class as the other determined by the ground truth,  so $P(T=j)=1/K$ for any $j$.  \\

From normalisation of probability,  we can write $P_{kk}=1-\sum_{j \neq k} P_{jk}$,  which allows us to rewrite $\sum_j \sum_{k\neq j}P_{jk}= \sum_{k} (\sum_{j \neq k}  P_{jk})=\sum_k (1-P_{kk})=K-\sum_j P_{jj}$.  Now inserting this into Eq.~\eqref{eq:tildeAjj} we find
\begin{align}
&\tilde{A}=\frac{1}{K}(A \sum_j P_{jj}+(1-A)(K-\sum_j P_{jj})=\frac{1}{K}(1-2A)\sum_j P_{jj}+(1-A)
\end{align}
Now using the generic assumption in the second line that the samples we have are `unbiased' also with respect to the mode,  so $P(C=j)=1/K$,  we see that 
\begin{align}
\sum_j P_{jj}=\sum_j\frac{P(\tilde{C}=C=j)}{P(C=j)}=KP(\tilde{C}=C)=KA^*.
\end{align}
Therefore,  we can write the general relation for \textit{any} noise type as 
\begin{align} \label{eq:AtildeA}
\tilde{A}=A^*(2A-1)+(1-A).
\end{align}
So now we see that the source of the decrease in accuracy $\tilde{A}$ due to noise itself is due \textit{only} to $A^*=P(\tilde{C}=C)$.  The rest of the accuracy dependence is on the accuracy in the noiseless case $A$ which is independent of noise.  \\

It is through the behaviour of $A^*$ that makes depolarisation noise quite special,  since in the \textit{infinite sampling} limit $N \rightarrow \infty$,  then $\tilde{C}_{\sigma}=C_{\sigma}$ for all states $\sigma$,  which gives $A^*=1$.  This was the result already proved in Lemma 1, in the manuscript  with details in Appendix A.  Therefore,  if we allow $N \rightarrow \infty$ sampling of the quantum circuit,  inserting $A^*=1$ into Eq.~\eqref{eq:AtildeA} gives
\begin{align}
\tilde{A}=(2A-1)+(1-A)=A,
\end{align}
so the accuracy in the absence of noise is the same as the accuracy in the presence of depolarisation noise! Note that this is special to depolarisation noise having the property $\tilde{C}=C$ and is not true for other noises.  However,  there are also other classes of noises that can help the model remain robust.  Some other robustness properties have been investigated in [44] for other types of noises and their impact on accuracy can be found through Eq.~\eqref{eq:AtildeA}. \\

However,  in the \textit{finite sampling} limit,  $A^* \neq 1$ for depolarisation noise.  We already proved in Lemma  2 of the manuscript that $A^*>1-2\exp(-2N \zeta^2)$ where $N$ is the number of samples we take from the quantum circuit and $\zeta$ is the precision to which we determine the quantity $\mathbf{\tilde{y}}_0(\sigma)$.  Therefore,  we see that in the finite sampling limit,  the accuracy in the presence of noise degrades as 
\begin{align}
\tilde{A}>(1-2\exp(-2N \zeta^2))(2A-1)+(1-A)
\end{align}
where $\tilde{A} \rightarrow A$ exponentially quickly as $N$ grows,  so accuracy in the presence of depolarisation noise is not in fact compromised very much. Thus,  all effects of depolarisation noise on accuracy can be remedied by taking more measurements (efficiently). 

\section*{Appendix B: Proof of Theorem 2} \label{app:B}
The argument runs in the same way as for Theorem 1. Here we look at $K$-class classification where the data is unbiased, but the model is biased. However, we look at the simplest setting where only one class is biased. Suppose it is class $A$ that is biased, so then $P(h=a)=\alpha$, and the rest of the classes are unbiased, so $P(h=i \neq a)=(1-\alpha)/(K-1) \equiv \alpha'$. In this case, we just need the formula for $\sum_j P_{jj}=P(\tilde{h}=1|h=1)+...+P(\tilde{h}=K|h=K)$, since in the multiclass case we have derived:
\begin{align}
    \tilde{A}=\frac{1}{K}(1-2A)\sum_j P_{jj}+(1-A). 
\end{align}
So we begin with
\begin{align}
    & A^*=P(\tilde{h}=h=1)+...+P(\tilde{h}=h=K) \nonumber \\
    &=\alpha P(\tilde{h}=a|h=a)+\alpha'\sum_{i \neq A}P(\tilde{h}=i|h=i) \nonumber \\
    &=\alpha \sum_j P_{jj}+(\alpha'-\alpha)(\sum_j P_{jj}-P(\tilde{h}=a|h=a)) \nonumber \\
    &=\alpha' \sum_j P_{jj}-(\alpha'/\alpha-1)P(\tilde{h}=h=a).
\end{align}
We can rearrange to find
\begin{align}
    \sum_{j} P_{jj}=\frac{1}{\alpha'}(A^*+(\alpha'/\alpha-1)P(\tilde{h}=h=a))
\end{align}
where $\alpha'=(1-\alpha)/(K-1)$. Finally we have the relation
\begin{align}
    \tilde{A}=\frac{(1-2A)(K-1)}{(1-\alpha)K}\left(A^*+\left(\frac{1-\alpha}{\alpha (K-1)}-1\right)P(\tilde{h}=h=a)\right)+(1-A),
\end{align}
which reduces to the Theorem 1 expression when $\alpha'=\alpha=1/K$.
\section*{Appendix C: Details of examples}
\subsection*{Example 1} 
We now have a noise model where we have the perturbation $x_2,...,x_{d+1} \rightarrow x'_2,...,x'_{d+1}$, where if  $(x_2,...,x_{d+1}) \sim D_{\pm}$ then $(x'_2,...,x'_{d+1}) \sim D_{\mp}$. \\

Under this noise model, we have the transformation
\begin{align}
    & \tilde{H}_1(x)=\frac{1}{d}(x'_2+...+x'_{d+1}) \nonumber \\
    & \tilde{H}_2(x)=H_2(x), 
\end{align}
where we see that $\tilde{H}_2=H_2$ but $\tilde{H}_{1} \neq H_{1}$. We can easily see that $A_2=p$ and $A^*_2=1$. Assuming that $P(c=1)=1/2=P(c=-1)$, $A_1$ can be computed using
\begin{align}
    A_1=\frac{1}{2}(P(h=1|c=1)+P(h=-1|c=-1)).
\end{align}
This accuracy can be increased to arbitrarily high values. For example, if $D_{\pm}$ are Gaussian distributions centered on $\pm \eta c$ for some $\eta>0$ with standard deviation $1$ and for large enough $d$ for the central limit theorem to apply, then $A_1=P(h=c)$ corresponds to the probability that the value of a Gaussian distribution centered on $\eta$ with standard deviation $1/d$ has a value larger than $0$. This means that so long as $\eta$ is larger than $1/\sqrt{d}$, we can have a high $A_1$.  We also find that
\begin{align}
   & A^*_1=P_{x_2,...x_{d+1}\sim D_+}(h_1(x)=\tilde{h}_1(x)) \nonumber \\
   &+P_{x_2,...x_{d+1}\sim D_-}(h_1(x)=\tilde{h}_1(x))=0
\end{align}
since we can imagine a perturbation where $x_{i\geq 2} \rightarrow -x_{i \geq 2}$ to give $D_{\pm} \rightarrow D_{\mp}$. This means that $h_1(x)=-\tilde{h}_1(x)$ for all $x$ selected from the distribution $D_+ \bigcup D_-$. To compute the robustness accuracy, we note that the noise perturbation makes the following change in distribution $D_{\pm} \rightarrow D_{\mp}$, which means that for the Gaussian distribution example we have $\tilde{A}_1=1-A_1$, which from Eq.~\eqref{eq:lemma1} is also consistent with $A^*_1=0$.\\

In the special case of the Gaussian distributions, we have
\begin{align}
    & A_1=1-\text{Erf}(\eta \sqrt{d}), A^*_1=0, \tilde{A}_1=1-A_1 \nonumber \\
    & A_2=p, A^*_2=1, \tilde{A}_2=A_2.
\end{align}
This means that we can have the trade-off scenario where $A_1>>A_2$ and $\tilde{A}_1=1-A_1<<1-A_2<A_2=\tilde{A}_2$ when $A_2=p>0.5$. Therefore, in the scenario $p>0.5$, we have a trade-off condition where a better model in accuracy $H_2 \rightarrow H_1$ can lead to a less robust model. 
\subsection*{Example 2} 
This means we have the following relationships between probabilities:
\begin{align}
    & p_{++} \equiv P_{\{x_1, x_2,...,x_{d+1}\}\sim D_+}  (f=1|g=1) \nonumber \\
    &=P_{\{x_1, x_2,...,x_{d+1}\}\sim D_+} (f=1|g=-1)\equiv p_{-+} \nonumber \\
    & p_{+-} \equiv P_{\{x_1, x_2,...,x_{d+1}\}\sim D_-}  (f=1|g=1)\nonumber \\
    &=P_{\{x_1, x_2,...,x_{d+1}\}\sim D_-} (f=1|g=-1) \equiv p_{--} \nonumber \\
    & p'_{++} \equiv P_{\{x_1, x_2,...,x_{d+1}\}\sim D_+}  (g=1|g=1) \nonumber \\
    &=P_{\{x_1, x_2,...,x_{d+1}\}\sim D_-} (g=1|g=1)\equiv p'_{+-}=1 \nonumber \\
     & p'_{-+} \equiv P_{\{x_1, x_2,...,x_{d+1}\}\sim D_+}  (g=1|g=-1)\nonumber \\
     &=P_{\{x_1, x_2,...,x_{d+1}\}\sim D_-} (g=1|g=-1)\equiv p'_{--}=0.
\end{align}
Then it is possible to compute
\begin{align}
    & A_1=1-\frac{1}{2}pa-\frac{1}{2}(1-p)b \nonumber \\
    & \tilde{A}_1=\frac{1}{2}(1-p)a+\frac{1}{2}pb,
\end{align}
where $a \equiv 1-p_{++}+p_{--}$ and $b \equiv 1-p_{-+}+p_{+-}$. From this we can prove that if $A_1>A_2=p$, then it is always true that $\tilde{A}_1<\tilde{A}_2$ if $p>1/2$. We see that $A_1>p$ implies $1/2pa<(1-p)(1-1/2b)$. $p>1/2$ is equivalent to $1-p<p$, which implies $1/2pa<p(1-1/2b)$. Inserting this inequality into $\tilde{A}_1=1/2(1-p)a+1/2pb<1/2pa+1/2pb<p(1-1/2b)+1/2pb=p=\tilde{A}_2$ gives our result. 
\subsection*{Example 3} 
If the distributions $D_{\pm}$ are the same Gaussian distributions in Example 1, then the argument proceeds in the same way, so that if $1-\text{Erf}(\eta \sqrt{d})>p>0.5$, then we have a trade-off situation that $A_1>A_2$ gives $\tilde{A}_1<\tilde{A}_2$. \\

This means that if we are training a model 
\begin{align}
     H(\sigma)=\sum_{i=1}^{d+1} \text{Tr}(h \sigma_i)=\alpha_1 H_1(\sigma)+\alpha_2 H_2(\sigma)
\end{align}
where we begin with $\alpha_1=0, \alpha_2=1$ and after more training we get $\alpha_0=1, \alpha_2=0$.
\subsection*{Example 4} 
Similarly to the classical scenario, we can make the following definitions for $j=1,2$
\begin{align}
     & p^{(j)}_{++} \equiv P_{\{f_2,...,f_{d+1}\}\sim D_+}  (H_j=1|H_1=1) \nonumber \\
    & p^{(j)}_{+-} \equiv P_{\{f_2,...,f_{d+1}\}\sim D_-}  (H_j=1|H_1=1)\nonumber \\
    & p^{(j)}_{-+} \equiv P_{\{f_2,...,f_{d+1}\}\sim D_+}  (H_j=1|H_1=-1) \nonumber \\
     & p^{(j)}_{--} \equiv P_{\{f_2,...,f_{d+1}\}\sim D_-}  (H_j=1|H_1=-1)
\end{align}
We then define the accuracy $A_1$ when $P(c=1)=1/2=P(c=-1)$ to be 
\begin{align}
    A_1=P(h_1=c) \equiv p.
\end{align}
The accuracy for the second model is then
\begin{align} 
A_2=1-\frac{1}{2}pa_2-\frac{1}{2}(1-p)b_2
\end{align}
where $a_2=1-p^{(2)}_{++}+p^{(2)}_{--}$ and $b_2=1-p^{(2)}_{-+}+p^{(2)}_{+-}$. Suppose we choose perturbations where only $f_2,...,f_{d+1}$ can change and leave $f_1$ invariant. Furthermore, the distributions change according to $D_{\pm} \rightarrow D_{\mp}$. \\

Suppose the perturbation leaves $H_1$ intact, so then 
\begin{align}
    \tilde{A}_1=A_1=p,
\end{align}
but model $H_2$ has a robustness accuracy different to its accuracy
\begin{align}
    \tilde{A}_2=\frac{1}{2}(1-p)a_2+\frac{1}{2}pb_2
\end{align}
Then a similar trade-off argument follows where if $p>1/2$, then if we train values of $\alpha_i$ beginning from $\alpha_1=1, \alpha_{i \neq 0}=0$ to more general $\{\alpha_i \}$ we can achieve
\begin{align}
    A_2>A_1
\end{align}
but this would lead to 
\begin{align}
    \tilde{A}_2<\tilde{A}_1.
\end{align} 
\subsection*{Example 7} 
A more general class of examples of trade-offs between $\tilde{A}_s$ and $A_s$ can come about if we consider small perturbations in the distribution $\mathcal{D}$ itself. Let the initial distribution be $\mathcal{D}_1$ and the distribution after perturbation be $\mathcal{D}_2$. We can define these distributions in the following way.\\

For example, suppose we have an observable $\mathcal{M}_j$ and we define $F_j(\sigma_i)=\text{Sign}\text{Tr}(\mathcal{M}_j\sigma_i)$ as the sign of the $j^{th}$ \textit{quantum feature} of the state $\sigma_i$. Then let us define the sign of one quantum feature $F_0$ such that 
\begin{align}
    \nonumber P(F_0(\sigma^{\pm}_i)=\pm 1|c(\sigma_i)=\pm 1)=f_0 \\
    \nonumber P(F_0(\sigma^{\mp}_i)=\pm 1|c(\sigma_i)=\mp 1)=1-f_0.
\end{align}
So if $f_0>1/2$, then the feature $F_0$ is positively correlated with the true label $c$. With this definition, we can write 
\begin{align}
    &\sum_{\sigma^+_i} P_s(h_s(\sigma^+_i)=1)=\sum_{\sigma_i}P_s(h_s(\sigma_i)|c_i) \nonumber \\
    &=\sum_{\sigma_i}(P_s(h_s(\sigma_i)=1|f_0(\sigma_i)=1)\mathbf{1}(f_0(\sigma_i)=c_i)\nonumber \\
    &+P_s(h_s(\sigma_i)=1|f_0(\sigma_i)=-1)\mathbf{1}(f_0(\sigma_i)=-c_i)).
\end{align}
This can be included in the definition of the accuracy where we define $\sigma^{f_0^{\pm}}=\{\sigma| f_0(\sigma)=\pm 1\}$.
\begin{align}
    &A_s=\sum_{\sigma_i}P(\sigma_i, c_i=1) ((\mathbf{1}(f_0(\sigma_i)=c_i)\text{Tr}(\Pi_1 \sigma^{f_0^+}_i)+\mathbf{1}(f_0(\sigma_i)=-c_i)\text{Tr}(\Pi_1\sigma^{f_0^{-}}_i))) \nonumber \\
    &+\sum_{\sigma_i}P(\sigma_i, c_i=-1) ((\mathbf{1}(f_0(\sigma_i)=c_i)\text{Tr}(\Pi_{-1} \sigma^{f_0^-}_i) +\mathbf{1}(f_0(\sigma_i)=-c_i)\text{Tr}(\Pi_{-1}\sigma^{f_0^{-}}_i))) \nonumber \\
    &=\sum_{\sigma_i^+}\text{Tr}(\Pi_1 \sigma^{f_0^+}_i)P(\sigma^+_i, f_0(\sigma_i)=c_i=1) +\sum_{\sigma_i^+}\text{Tr}(\Pi_1 \sigma^{f_0^-}_i)P(\sigma^+_i, f_0(\sigma_i)=-c_i=-1) \nonumber \\
    &+\sum_{\sigma_i^-}\text{Tr}(\Pi_{-1} \sigma^{f_0^-}_i)P(\sigma^-_i, f_0(\sigma_i)=c_i=-1) +\sum_{\sigma_i^-}\text{Tr}(\Pi_{-1} \sigma^{f_0^+}_i)P(\sigma^-_i, f_0(\sigma_i)=-c_i=1).
\end{align}
We observe that
\begin{align}
   & P(\sigma^{\pm}_i,f_0(\sigma_i)=\pm 1)=P(f_0(\sigma_i)=\pm 1|\sigma^{\pm}_i)=f_0 P(\sigma^{\pm}_i) \nonumber \\
   & P(\sigma^{\pm}_i,f_0(\sigma_i)=\mp 1)=P(f_0(\sigma_i)=\mp 1|\sigma^{\pm}_i)=(1-f_0) P(\sigma^{\pm}_i).
\end{align}
Inserting this into the above expression for $A_s$ gives
\begin{align}
    &A_s=\sum_{\sigma_i^+}\text{Tr}(\Pi_1 \sigma^{f_0^+}_i)P(\sigma^+_i)f_0+\sum_{\sigma_i^+}\text{Tr}(\Pi_1 \sigma^{f_0^-}_i)P(\sigma^+_i)(1-f_0) \nonumber \\
    &+\sum_{\sigma_i^-}(1-\text{Tr}(\Pi_{1} \sigma^{f_0^-}_i))P(\sigma^-_i)f_0+ \sum_{\sigma_i^-}(1-\text{Tr}(\Pi_{1} \sigma^{f_0^+}_i))P(\sigma^-_i)(1-f_0).
\end{align}
We can define the normalised states
\begin{align}
    & \Sigma^{++}=\frac{\sum_{\sigma_i^+} P(\sigma_i^+)\sigma^{f_0^+}_i}{\sum_{\sigma_j^+} P(\sigma_j^+)} \nonumber \\
    &\Sigma^{+-}=\frac{\sum_{\sigma_i^+} P(\sigma_i^+)\sigma^{f_0^-}_i}{\sum_{\sigma_j^+} P(\sigma_j^+)} \nonumber \\
    &\Sigma^{-+}=\frac{\sum_{\sigma_i^-} P(\sigma_i^-)\sigma^{f_0^+}_i}{\sum_{\sigma_j^-} P(\sigma_j^-)} \nonumber \\
     &\Sigma^{--}=\frac{\sum_{\sigma_i^-} P(\sigma_i^-)\sigma^{f_0^-}_i}{\sum_{\sigma_j^-} P(\sigma_j^-)} \nonumber \\
\end{align}
If we use the unbiased dataset assumption, then we can simplify
\begin{align}
  &  A_s=\frac{1}{2}f_0 \text{Tr}(\Pi_1 \Sigma^{++})+\frac{1}{2}(1-f_0)\text{Tr}(\Pi_1\Sigma^{+-}) \nonumber \\
&+\frac{1}{2}f_0(1-\text{Tr}(\Pi_1 \Sigma^{--}))+\frac{1}{2}(1-f_0)(1-\text{Tr}(\Pi_1 \Sigma^{-+})).
\end{align}
Then a perturbed classifier leads to 
\begin{align}
    & \tilde{A}_s=\frac{1}{2}f_0 \text{Tr}(\Pi_1 \tilde{\Sigma}^{++})+\frac{1}{2}(1-f_0)\text{Tr}(\Pi_1\tilde{\Sigma}^{+-}) \nonumber \\
&+\frac{1}{2}f_0(1-\text{Tr}(\Pi_1 \tilde{\Sigma}^{--}))+\frac{1}{2}(1-f_0)(1-\text{Tr}(\Pi_1 \tilde{\Sigma}^{-+})).
\end{align}
Now we study a class of perturbations such that 
\begin{align}
    & \tilde{\Sigma}^{++}=\Sigma^{+-} \nonumber \\
    & \tilde{\Sigma}^{+-}=\Sigma^{++} \nonumber \\
    & \tilde{\Sigma}^{--}=\Sigma^{-+} \nonumber \\
    &\tilde{\Sigma}^{-+}=\Sigma^{--}.
\end{align}
Then it is straightforward to show that for any $f_0>1/2$ we have the trade-off condition
\begin{align}
    \tilde{A}_s \leq \frac{f_0}{1-f_0}(1-A_s).
\end{align} 
\section*{Appendix D: Proof of Lemma 1}
Below we will distinguish between two types of probabilities. $P_s$ is the probability associated with the stochastic variable getting a specific outcome given some $\sigma$, which we can model with a quantum expectation value as $P_s(h_s(\sigma)=j)=\text{Tr}(\Pi_j \sigma)$. On the other hand $P$ is the probability associated with averaging over some distribution of input states, for instance $P(c(\sigma)=1)$ is the probability of states sampled from $\mathcal{D}$ as having the true label $1$. We can define then accuracy $A_s$ of this stochastic classifier $h_s$ as the probability it gives the same outcome as the true label
\begin{align}
    A_s \equiv \sum_{\sigma_i}\sum_{c_i}P(\sigma_i, c_i)P_s(h_s(\sigma_i)=c_i)
\end{align}
Using the definition $\sigma^{\pm}=\{\sigma| c(\sigma)=\pm 1\}$ we can then rewrite $A_s$ 
\begin{align}
    &A_s= \sum_{\sigma_i}\sum_{c_i}P(\sigma_i, c_i)P_s(h_s(\sigma_i)=c_i) \nonumber \\
    &=\sum_{\sigma_i^+}P(\sigma^+_i, c_i=1)P_s(h_s(\sigma^+_i)=1) +\sum_{\sigma_i^-}P(\sigma^-_i, c_i=-1)P_s(h_s(\sigma^-_i)=-1) \nonumber \\
    &=\sum_{\sigma_i^+}P(\sigma^+_i, c_i=1) \text{Tr}(\Pi_1 \sigma_i^+) +\sum_{\sigma_i^-}P(\sigma^-_i, c_i=-1) \text{Tr}(\Pi_{-1} \sigma_i^-) \nonumber \\
    &=\text{Tr}(\Pi_1\sum_{\sigma_i^+}P(\sigma^+_i)\sigma_i^+) +\text{Tr}(\Pi_{-1}\sum_{\sigma_i^-}P(\sigma^-_i)\sigma_i^-) \nonumber \\
    &=(\sum_{\sigma_i^+}P(\sigma^+_i))\text{Tr}(\Pi_1 \Sigma^+)+(\sum_{\sigma_i^-}P(\sigma^-_i))\text{Tr}(\Pi_{-1}\Sigma^-) \nonumber \\
    &
\end{align}
where in the third equality we used $P(\sigma^{\pm}_i, c_i=\pm 1)=P(c_i=-\pm 1|\sigma^{\pm}_i)P(\sigma^{\pm}_i)=P(\sigma^{\pm}_i)$ and we define the normalised states 
\begin{align}
     \Sigma^{\pm} \equiv \frac{\sum_{\sigma_i^{\pm}}P(\sigma^{\pm}_i)\sigma_i^{\pm}}{\sum_{\sigma_j^{\pm}}P(\sigma^{\pm}_j)}.
\end{align}
For a unbiased dataset we have $\sum_{\sigma_j^{\pm}}P(\sigma^{\pm}_j)=1/2$. 
Then using the relation $\Pi_1=\mathbf{1}-\Pi_{-1}$ we can write 
\begin{align}
    A_s=\frac{1}{2}\text{Tr}(\Pi_1(\Sigma^+-\Sigma^-))+\frac{1}{2}.
\end{align}
We will work using this unbiased assumption unless otherwise stated. Suppose we modify only the stochastic classifier, so we change the measurement $\Pi \rightarrow \tilde{\Pi}=\sum_l E_l \Pi E_l^{\dagger}$ where $\{E_l\}$ are Kraus operators. Then the robustness accuracy under this perturbation becomes 
\begin{align}
   &\tilde{A}_s=\frac{1}{2}\text{Tr}(\tilde{\Pi}_1(\Sigma^+-\Sigma^-))+\frac{1}{2} \nonumber \\
   &=\frac{1}{2}\text{Tr}(\Pi_1(\tilde{\Sigma}^+-\tilde{\Sigma}^-))+\frac{1}{2}
\end{align}
where $\tilde{\Sigma}^{\pm}=\sum_l E^{\dagger}_l \Sigma^{\pm}E_l$. \\

Then we can derive the difference
\begin{align} \label{eq:asdifference1}
    |\tilde{A}_s-A_s|=\frac{1}{2}|\text{Tr}(\Pi_1(\tilde{\Sigma}^+-\Sigma^++\Sigma^--\tilde{\Sigma}^-)|.
\end{align}
We can rearrange Eq.~\eqref{eq:asdifference1} to get 
\begin{align}
    & |\tilde{A}_s-A_s|=\frac{1}{2}|\text{Tr}(\Pi_1(\tilde{\Sigma}^++\Sigma^--(\Sigma^++\tilde{\Sigma}^-)| \nonumber \\
   & =|\text{Tr}(\Pi_1 (\Sigma_A-\Sigma_B))| \leq \tau(\Sigma_A, \Sigma_B)
\end{align}
where $\Sigma_A \equiv (1/2)(\tilde{\Sigma}^++\Sigma^-)$ and $\Sigma_B \equiv (1/2)(\Sigma^++\tilde{\Sigma}^-)$. As a very quick check, when the noise is the identity operation, we see that we see that the upper bound is saturated, since $\tau(\Sigma_A, \Sigma_B=\Sigma_A)=0$.

\bibliographystyle{IEEEtran}
\bibliography{Ref}

\begin{thebibliography}{1}
\providecommand{\url}[1]{#1}
\csname url@samestyle\endcsname
\providecommand{\newblock}{\relax}
\providecommand{\bibinfo}[2]{#2}
\providecommand{\BIBentrySTDinterwordspacing}{\spaceskip=0pt\relax}
\providecommand{\BIBentryALTinterwordstretchfactor}{4}
\providecommand{\BIBentryALTinterwordspacing}{\spaceskip=\fontdimen2\font plus
\BIBentryALTinterwordstretchfactor\fontdimen3\font minus \fontdimen4\font\relax}
\providecommand{\BIBforeignlanguage}[2]{{%
\expandafter\ifx\csname l@#1\endcsname\relax
\typeout{** WARNING: IEEEtran.bst: No hyphenation pattern has been}%
\typeout{** loaded for the language `#1'. Using the pattern for}%
\typeout{** the default language instead.}%
\else
\language=\csname l@#1\endcsname
\fi
#2}}
\providecommand{\BIBdecl}{\relax}
\BIBdecl

\bibitem{dobriban2023provable}
E.~Dobriban, H.~Hassani, D.~Hong, and A.~Robey, ``Provable tradeoffs in adversarially robust classification,'' \emph{IEEE Transactions on Information Theory}, vol.~69, no.~12, pp. 7793--7822, 2023.

\bibitem{tsipras2018robustness}
D.~Tsipras, S.~Santurkar, L.~Engstrom, A.~Turner, and A.~Madry, ``Robustness may be at odds with accuracy,'' \emph{arXiv preprint arXiv:1805.12152}, 2018.

\bibitem{xu2012robustness}
H.~Xu and S.~Mannor, ``Robustness and generalization,'' \emph{Machine learning}, vol.~86, pp. 391--423, 2012.

\bibitem{yang2020closer}
Y.-Y. Yang, C.~Rashtchian, H.~Zhang, R.~R. Salakhutdinov, and K.~Chaudhuri, ``A closer look at accuracy vs. robustness,'' \emph{Advances in neural information processing systems}, vol.~33, pp. 8588--8601, 2020.

\bibitem{zhang2019theoretically}
H.~Zhang, Y.~Yu, J.~Jiao, E.~Xing, L.~El~Ghaoui, and M.~Jordan, ``Theoretically principled trade-off between robustness and accuracy,'' in \emph{International conference on machine learning}.\hskip 1em plus 0.5em minus 0.4em\relax PMLR, 2019, pp. 7472--7482.

\bibitem{yin2019fourier}
D.~Yin, R.~Gontijo~Lopes, J.~Shlens, E.~D. Cubuk, and J.~Gilmer, ``A fourier perspective on model robustness in computer vision,'' \emph{Advances in Neural Information Processing Systems}, vol.~32, 2019.

\end{thebibliography}

\end{document}